\documentclass[sigconf]{aamas}

\usepackage{balance} % for balancing columns on the final page

\makeatletter
\gdef\@copyrightpermission{
  \begin{minipage}{0.2\columnwidth}
   \href{https://creativecommons.org/licenses/by/4.0/}{\includegraphics[width=0.90\textwidth]{by}}
  \end{minipage}\hfill
  \begin{minipage}{0.8\columnwidth}
   \href{https://creativecommons.org/licenses/by/4.0/}{This work is licensed under a Creative Commons Attribution International 4.0 License.}
  \end{minipage}
  \vspace{5pt}
}
\makeatother

\setcopyright{ifaamas}
\acmConference[AAMAS 2027]{Proc.\@ of the 26th International Conference
on Autonomous Agents and Multiagent Systems (AAMAS 2027)}{3--7 May 2027}{Hanoi, Vietnam}{M.~Baldoni, F.~Fang, W.~Yeoh, N.~Yorke-Smith (eds.)}
\copyrightyear{2027}
\acmYear{2027}
\acmDOI{}
\acmISBN{}

\acmSubmissionID{1637}

\submissionType{Research Paper Track}
\title{The Delegation Dilemma: Economies of Scale, Accountability, and Polycentric Governance}

\author{Arin Mukhopadhyay}
\affiliation{
  \institution{University of Edinburgh}
  \city{Edinburgh}
  \country{United Kingdom}}
\email{a.mukhopadhyay-1@sms.ed.ac.uk}

\author{Aris Filos-Ratsikas}
\affiliation{
  \institution{Imperial College London}
  \city{London}
  \country{United Kingdom}}
\email{aris.filos-ratsikas@imperial.ac.uk}

\author{Panayiotis Danassis}
\affiliation{
  \institution{University of Southampton}
  \city{Southampton}
  \country{United Kingdom}}
\email{p.danassis@soton.ac.uk}

\begin{abstract}
Many collective-action systems are neither fully decentralized nor fully centralized: agents voluntarily organize into institutions that share organizational costs and delegate decisions to representatives. Delegation can reduce individual costs, but it also transfers control to representatives whose private objectives need not align with those of their constituents. We study this tension in a repeated game with endogenous delegation to self-interested centers. A simple reduced-form model yields a \emph{scale--accountability tradeoff}: larger constituencies offer greater cost-sharing benefits to members, yet the marginal value to a center of retaining any one constituent decreases as group size grows. Consequently, the prospect of individual exit can sustain non-exploitative delegation only up to an endogenous accountable constituency size. Additionally, we identify parameter regions in which fully delegated, accountable governance, when feasible, must be \emph{polycentric}. We investigate these predictions through simulations with large language model (LLM) agents. Across five matched pairs of 30-agent simulations, polycentric structures repeatedly emerge but remain fluid, exhibiting cycles of concentration, exploitation, exit, and re-formation. We further find that models of different capacity exhibit similar willingness to assume leadership but differ substantially in their exercise of authority, with Sol initiating exploitation more frequently than Luna in matched decision states.
\end{abstract}


\usepackage{amsmath,mathtools}
\usepackage{amsthm}
\usepackage{booktabs}
\usepackage{multirow}
\usepackage{array}
\usepackage{enumitem}
\usepackage{tikz}
\usepackage{xspace}
\usepackage{listings}

\lstdefinestyle{agentprompt}{
  basicstyle=\ttfamily\scriptsize,
  breaklines=true,
  breakatwhitespace=true,
  columns=fullflexible,
  keepspaces=true,
  showstringspaces=false,
  frame=single,
  framerule=0.2pt,
  aboveskip=5pt,
  belowskip=5pt,
  xleftmargin=2pt,
  xrightmargin=2pt
}

\newtheorem{definition}{Definition}
\newtheorem{assumption}{Assumption}
\newtheorem{lemma}{Lemma}
\newtheorem{proposition}{Proposition}
\newtheorem{theorem}{Theorem}
\newtheorem{corollary}{Corollary}
\newcommand{\mmin}{m_{\min}}
\newcommand{\mmax}{m_{\max}}
\newcommand{\mbar}{\overline m}
\newcommand{\munder}{\underline m}

\usepackage[capitalise,noabbrev]{cleveref}

\definecolor{BrightBlue}{RGB}{65, 145, 225}

\makeatletter
\newcommand\footnoteref[1]{\protected@xdef\@thefnmark{\ref{#1}}\@footnotemark}
\makeatother

\renewcommand\footnotetextcopyrightpermission[1]{} % removes footnote with conference information in first column
\begin{document}

%%% The following commands remove the headers in your paper. For final 
%%% papers, these will be inserted during the pagination process.

\pagestyle{fancy}
\fancyhead{}

%%% The next command prints the information defined in the preamble.

\maketitle 

%%%%%%%%%%%%%%%%%%%%%%%%%%%%%%%%%%%%%%%%%%%%%%%%%%%%%%%%%%%%%%%%%%%%%%%%

\begin{figure}[t]
    \centering
    \includegraphics[width=\linewidth]{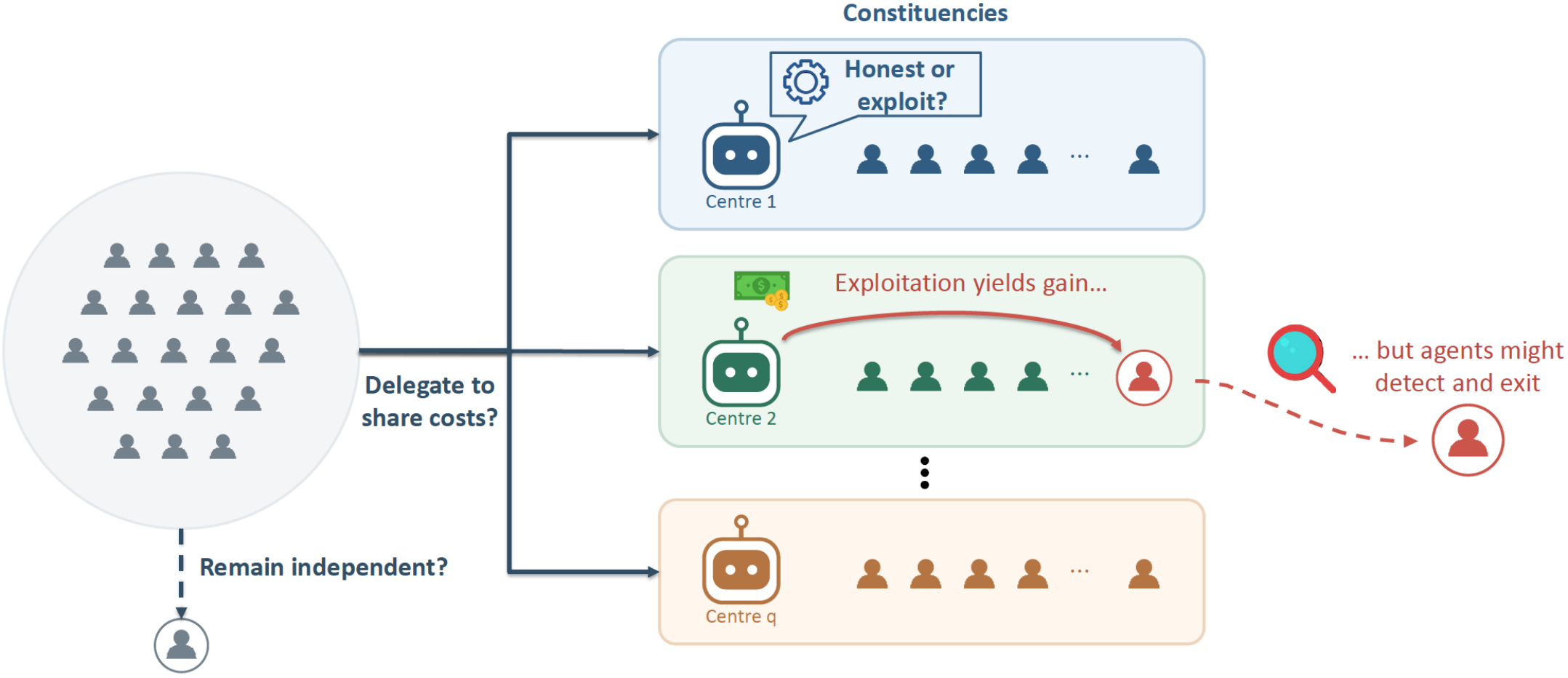}
    \caption{Overview of our delegation setting. Agents choose whether to remain independent or join a constituency (to share costs) led by a self-interested center. Each center decides whether to exploit its delegates for private gain. Delegates may detect exploitation and exit the delegation.}
    \label{fig:delegation-game}
\end{figure}

\section{Introduction}

Socially valuable activities often exhibit \emph{economies of scale}: they are costly, or even infeasible, when undertaken individually, but become easier to perform as individuals organize into groups. Farmers may cooperate in production or irrigation, fishers may jointly manage a common resource, and communities may pool resources to provide local public services. Such activities naturally give rise to constituencies, with organizers or leaders who coordinate contributions, make collective decisions, and determine how the gains from cooperation are distributed. The benefits of scale have long been central to theories of clubs and collective action \cite{buchanan1965clubs,ostrom1961organization}, as well as to more recent work on endogenous group formation \cite{page2005voluntary,ahn2008endogenous}. 

At the same time, organizational scale may alter the incentives of those entrusted with authority. As a constituency grows, the centre's marginal organizational benefit from retaining any particular constituent may decline. A self-interested centre may therefore find it profitable to exploit some delegates, even when doing so risks their eventual departure. When losing a constituent has only a small effect on the centre's future organizational costs, the threat of exit provides a weaker deterrent against exploitation. Thus, the same economies of scale that encourage collective organization can also undermine the accountability of those exercising delegated authority and create scope for opportunism. This tension echoes the classical observation that organizational growth can concentrate decision-making authority \cite{michels1915political}.

However, such opportunism does not necessarily imply the persistence of a single governing authority. Constituents who receive an exploitative treatment can instead exercise \emph{exit} \cite{hirschman1970exit}, by joining an alternative constituency or by volunteering to lead a new constituency themselves. Competition for constituents therefore constraints the size of the groups and their leaders’ ability to benefit at the expense of their members. At the same time, groups cannot become arbitrarily small, because doing so sacrifices the economies of collective production. The resulting institutional structure is therefore neither individual nor fully centralized, but \emph{polycentric}: multiple, relatively autonomous centers of organization coexist, compete for constituents, and may cooperate where doing so is mutually beneficial. This conception is closely related to the polycentric political systems of Ostrom, Tiebout, and Warren \cite{ostrom1961organization}, and to Ostrom’s later account of polycentric governance as an alternative to the dichotomy between centralized state provision and purely decentralized market exchange \cite{ostrom1990governing}. Examples of polycentricity are also prevalent in practice. For example, in experiments with members of Chilean territorial-use-right fisheries, Chávez et al. \cite{chavez2023cprcoalitions} find that while fishers frequently formed the grand coalition when monitoring was inexpensive, in fully decentralized settings without monitoring, the former coalitions were many more and significantly smaller. Thus, our general aim in this paper is to study and understand the following phenomenon: polycentricity arises not merely because different tasks are efficiently performed at different scales, but because organizational scale itself changes the incentives of those entrusted with authority.

With the rapid development and employment of AI systems, individuals and organizations increasingly delegate their decisions to autonomous AI agents, and in particular to agents powered by large language models that act on their behalf. A growing literature has already begun to examine delegation to artificial agents in collective-action environments, e.g., see \cite{fernandezdomingos2022delegation,nguyen2025navigating,faulkner2026leadership}. This raises the question of whether the institutional forces that generate polycentricity among human decision makers continue to operate when such decisions are delegated to AI agents. Will AI agents recognize the economic advantages of delegation and voluntarily form constituencies? How will self-interested AI centres exercise their authority as constituency size increases? Can the prospect of losing constituents discipline exploitation, or will centres anticipate that departing members can be replaced? Will agents establish competing centres when existing governance arrangements become unattractive? Finally, will different language models respond similarly to the same institutional incentives, or exhibit distinct patterns in their exercise of delegated authority?

More generally, we are interested to know whether interactions among autonomous AI agents reflect the endogenous checks on organizational scale that sustain polycentricity among human decision makers, or whether they instead generate qualitatively different organizational structures. 
Our aim is therefore not to assume that polycentricity is a stable outcome, but to understand the institutional incentives that make it possible and the strategic behaviour that shapes its emergence and persistence.

\subsection{Our Contributions}

We study whether polycentric arrangements emerge among autonomous LLM agents, how they evolve as constituencies gain and lose members, and whether different models exercise delegated authority in similar ways. Our main contributions are as follows:

\textbf{(i) We develop a general model of voluntary delegation, economies of scale, and exit-based accountability} (Figure~\ref{fig:delegation-game}). Our model is deliberately stylized, and intentionally abstracts away from the particulars of the task (e.g., whether it is a resource harvesting task or any other task) and the form of exploitation (e.g., whether that arises via some misallocation of resources etc). The aim of the model is to capture the fundamental tension between the efficiency gains from forming larger groups and the centers' incentives for exploitation, and the conditions under which this leads to polycentric structures. Participation depends on the average-cost saving created by joining a constituency, whereas accountability depends on the marginal organizational value that the center loses when one constituent exits. We derive a minimum viable constituency size, a maximum size at which probabilistic individual exit can deter exploitation, comparative statics, and an exact characterization of the number of centers compatible with fully delegated accountable governance. A fixed-cost specialization yields closed-form thresholds. Our theoretical results serve as baselines to ground our simulations. 

\textbf{(ii) We study endogenous governance formation among Large Language Model (LLM) agents}. Across ten matched episodes with 30 agents, polycentric structures repeatedly emerge through both bottom-up formation and fragmentation of concentrated authority, but remain fluid as fee sharing attracts agents toward larger constituencies and exploitation induces exit. We additionally compare GPT-6 Luna and GPT-6 Sol (small vs. large model) pointwise on the same post-initialisation decision states. The two models are similarly willing to become centers but differ substantially in how they exercise delegated authority, with Sol initiating exploitation more often under matched conditions.

%%%%%%%%%%%%%%%%%%%%%%%%%%%%%%%%%%%%%%%%%%%%%%%%%%%%%%%%%%%%%%%%%%%%%%%%%%%%%%%%%%%%%%%%%%%
%%%%%%%%%%%%%%%%%%%%%%%%%%%%%%%%%%%%%%%%%%%%%%%%%%%%%%%%%%%%%%%%%%%%%%%%%%%%%%%%%%%%%%%%%%%

\section{Related Work} \label{sec:related}

Our model connects work on club and coalition formation with delegation and institutional accountability. Classical club theory studies how cost sharing and congestion determine desirable group size \cite{buchanan1965clubs}, while coalition-formation models study the endogenous emergence and stability of group structures \cite{rayvohra1999coalition,page2005voluntary,ahn2008endogenous}. Most closely related to our dynamic setting,~\citet{joshi2025fear} studies sequential club formation when agents' preferences depend on club size. Endogenous group formation has also been studied in collective-resource settings \cite{chavez2023cprcoalitions} and in environments with multiple centrally coordinated groups \cite{harlalka2023dining,harlalka2025multigroup}. Our upper size limit has a different source: larger groups dilute the disciplinary effect of losing any particular member rather than generating congestion or technological inefficiency.

A second line of literature studies delegation to representatives or intermediaries.~\citet{hamman2011electoral} show experimentally that groups may voluntarily delegate public-good decisions to elected members, and that poorly performing representatives can subsequently be replaced. Related work studies representation in public-good games \cite{hauge2015representing,hamman2010selfinterest}, delegation as a coordination mechanism \cite{corazzini2020delegation}, and delegation of common-pool-resource decisions to a stakeholder \cite{bayle2024delegation}. Our distinctive mechanism combines voluntary constituency formation with a self-interested center whose incentive to retain each constituent decreases with organizational scale. This links economies of scale directly to exit-based accountability.

Finally, recent work studies institutional mechanisms among LLM agents.~\citet{faulkner2026leadership} introduce elected leadership in LLM social groups, while~\citet{tewolde2026coopeval} compare repetition, reputation, mediation, and contracting as cooperation-sustaining mechanisms. Hierarchical authority, elections, and oversight have also been studied in LLM public-good environments \cite{seyedin2026hierarchical}. Rather than imposing a fixed governance institution, we ask whether self-interested AI agents endogenously form centralized, polycentric, or individual governance structures when organizational scale simultaneously creates economic benefits and weakens accountability.
%%%%%%%%%%%%%%%%%%%%%%%%%%%%%%%%%%%%%%%%%%%%%%%%%%%%%%%%%%%%%%%%%%%%%%%%%%%%%%%%%%%%%%%%%%%
%%%%%%%%%%%%%%%%%%%%%%%%%%%%%%%%%%%%%%%%%%%%%%%%%%%%%%%%%%%%%%%%%%%%%%%%%%%%%%%%%%%%%%%%%%%

\section{Delegation -- Accountability Model} \label{sec:theory-model}

In this section, we formulate a reduced-form game that isolates the institutional incentives associated with delegation, cost sharing, and exit. \textbf{All proofs can be found in Appendix~\ref{app: Proofs}.}

There are $n\ge2$ persistent agents. At any governance epoch, agents may remain under individual control or belong to a \emph{constituency} consisting of one center and the agents who delegate control to it. Its size is $m\ge 1$, including the center; $m=1$ corresponds to individual control. Delegation is binding for $L\ge1$ operational rounds. After each nonterminal operational round, play continues with probability $\delta\in(0,1)$, so a payoff $s$ rounds in the future is reached with probability $\delta^s$. 

A constituency of size $m$ incurs a per-round organizational cost $K(m)>0$, shared equally among its members. Let
\begin{equation}
c(m)=\frac{K(m)}{m}
\label{eq:per-member-cost}
\end{equation}
denote the resulting per-member cost. This formulation can represent shared infrastructure, administration, transaction costs, access or licensing costs, monitoring systems, or other technologies whose average cost changes with organizational scale. Delegates additionally incur a private per-round loss $\ell\ge0$ from surrendering direct control. A center does not incur $\ell$, because it retains decision authority. If the center exploits $k$ delegates in a round, it receives $kd$ and each targeted delegate loses $d$, where $d>0$. Thus the reduced-form per-round payoffs are
\begin{equation}
u_C(m,k)=-c(m)+kd,
\qquad
u_D(m,z)=-c(m)-\ell-dz,
\label{eq:reduced-payoffs}
\end{equation}
where $z\in\{0,1\}$ indicates whether a given delegate is exploited. An independent agent receives $u_I=-c(1)$. Exploitation is therefore a transfer between the center and a targeted delegate; its disciplinary effect operates through subsequent governance choices.

\begin{assumption}[\textbf{Economies of scale}]
\label{ass:scale}
The per-member organizational cost $c(m)$ is weakly decreasing in $m$. Define the center's one-member retention value
\begin{equation}
D(m)=c(m-1)-c(m), \qquad m\ge2.
\label{eq:Dgeneral}
\end{equation}
We additionally assume $D(m)$ is weakly decreasing in $m$, so the marginal organizational benefit of retaining one more constituent diminishes with scale.
\end{assumption}

The first part of Assumption~\ref{ass:scale} captures average-cost economies of scale. The second separates the mechanism relevant for accountability: as constituencies grow, \emph{losing any particular constituent has a smaller effect on the center's own organizational cost.}

\subsection{Delegation and Economies of Scale}

Relative to individual control, membership in a constituency of size $m$ saves the delegate
\begin{equation}
S(m)=c(1)-c(m)
\label{eq:saving}
\end{equation}
per round in organizational cost. The net per-round advantage of non-exploitative delegation is therefore
\begin{equation}
B(m)=c(1)-c(m)-\ell.
\label{eq:participation-benefit}
\end{equation}
Delegation is individually rational when $B(m)\ge0$.

\begin{proposition}[\textbf{General participation threshold}]
\label{prop:general-mmin}
Under Assumption~\ref{ass:scale}, if there exists a nontrivial constituency size for which delegation is individually rational, then the set of individually rational sizes is upward closed and its minimum is
\begin{equation}
\boxed{\mmin=\min\left\{m\ge2:c(1)-c(m)\ge\ell\right\}}
\label{eq:mmin-general}
\end{equation}
If this set is empty, no finite constituency makes delegation individually rational through organizational cost sharing alone.
\end{proposition}

\paragraph{\textbf{Fixed-cost Specialization}} 

To obtain closed-form results, we specialize to the case in which the total organizational cost is independent of constituency size:
\begin{equation}
K(m)=\kappa, \qquad \kappa>0,
\label{eq:fixed-cost}
\end{equation}
so that $c(m)=\kappa/m$. This specification captures a fixed cost that can be shared among constituency members and is the environment implemented in our simulations.

\begin{corollary}[\textbf{Minimum viable constituency size under fixed costs}]
\label{prop:mmin}
Suppose $K(m)=\kappa$ and $0\le\ell<\kappa$. Non-exploitative delegation is individually rational if and only if
\begin{equation}
m\ge \frac{\kappa}{\kappa-\ell}
\end{equation}
The minimum nontrivial integer constituency size is
\begin{equation}
\boxed{\mmin=\max\left\{2,\left\lceil\frac{\kappa}{\kappa-\ell}\right\rceil\right\}}
\label{eq:mmin}
\end{equation}
If $\ell\ge\kappa$, no finite constituency makes delegation individually rational through fixed-cost sharing alone.
\end{corollary}

\subsection{When Does Exit Discipline the Center?}
We next model the center's temptation to exploit a constituent. The parameter $d$ (defined in Eq.~\ref{eq:reduced-payoffs}) is deliberately reduced form: it is the center's per-round private gain from exploiting one delegate and can represent appropriation, preferential allocation, extraction of rents, or any other benefit obtained by using delegated authority against a constituent's interest.

\begin{assumption}[\textbf{Targeted exploitation and probabilistic exit}]
\label{ass:deviation}
At the beginning of a governance epoch, a center in a constituency of size $m\ge2$ can exploit one delegate, obtaining an additional payoff $d>0$ in each of the $L$ locked-in rounds while imposing the same per-round loss on that delegate. At the next governance opportunity, the exploitation is detected and results in that delegate's departure with probability $\alpha\in[0,1]$. Conditional on departure, the delegate is not replaced over the continuation horizon, so the change in the center's per-member organizational cost persists. The deviation does not otherwise alter the continuation probability.
\end{assumption}

The parameter $\alpha$ is a reduced-form accountability probability. If detection occurs with probability $p_{\mathrm{det}}$ and a delegate exits conditional on detection with probability $p_{\mathrm{exit}\mid\mathrm{det}}$, then $\alpha=p_{\mathrm{det}}p_{\mathrm{exit}\mid\mathrm{det}}$. This lets the theory encompass imperfect observability as well as imperfect behavioural response without changing the core incentive comparison.

\begin{definition}[\textbf{Accountable constituency}]
A constituency of size $m$ is \emph{accountable under individual exit} if, under Assumption~\ref{ass:deviation}, the center weakly prefers non-exploitative behaviour to the targeted exploitation deviation.
\end{definition}

\begin{proposition}[\textbf{General accountability condition}]
\label{prop:general-accountability}
Under Assumptions~\ref{ass:scale}--\ref{ass:deviation}, a constituency of size $m\ge2$ is accountable under individual exit if and only if
\begin{equation}
\boxed{\alpha\delta^L D(m)\ge d(1-\delta^L)}
\label{eq:IC-general}
\end{equation}
Because $D(m)$ is weakly decreasing, the set of accountable constituency sizes is downward closed. When the largest accountable size is finite, define
\begin{equation}
\mmax=\max\left(\{1\}\cup\left\{m\ge2:\alpha\delta^L D(m)\ge d(1-\delta^L)\right\}\right).
\label{eq:mmax-general}
\end{equation}
\end{proposition}

The general condition separates two distinct effects of scale. Participation depends on the \emph{average} organizational saving $c(1)-c(m)$, whereas accountability depends on the \emph{marginal} retention value $D(m)=c(m-1)-c(m)$. Organizational scale can therefore make delegation increasingly attractive to constituents while simultaneously making each individual constituent less valuable to the center.

\begin{corollary}[\textbf{Maximum accountable constituency size under fixed costs}]
\label{prop:mmax}
Under $K(m)=\kappa$, condition~\eqref{eq:IC-general} becomes
\begin{equation}
\alpha\delta^L\frac{\kappa}{m(m-1)}\ge d(1-\delta^L),
\label{eq:IC}
\end{equation}
or equivalently
\begin{equation}
m(m-1)\le\frac{\alpha\kappa\delta^L}{d(1-\delta^L)}.
\label{eq:ICquadratic}
\end{equation}
Hence the continuous accountability threshold is
\begin{equation}
\mbar=\frac{1+\sqrt{1+4\frac{\alpha\kappa\delta^L}{d(1-\delta^L)}}}{2},
\label{eq:mbar}
\end{equation}
and the largest integer constituency size supportable by individual exit is
\begin{equation}
\boxed{\mmax=\left\lfloor\mbar\right\rfloor}
\label{eq:mmax}
\end{equation}
If $\mmax=1$, no nontrivial delegated constituency is accountable through this exit mechanism alone.
\end{corollary}

\paragraph{\textbf{Interpretation}}

The upper threshold does not imply that constituencies larger than $\mmax$ cannot exist or that their constituents necessarily leave. It says only that, under the stated deviation and punishment technology, the expected loss generated by one constituent's possible future exit is insufficient to make the center prefer non-exploitative behaviour. Constituents may still prefer membership despite exploitation, and replacement, collective exit, reputation, auditing, legal sanctions, or other mechanisms can change the effective accountability constraint, sustaining accountability beyond $\mmax$.

\subsubsection{\textbf{Robustness of the Accountability Constraint}}

Two simple observations strengthen the interpretation of $\mmax$.

\begin{lemma}[\textbf{Single-constituent exploitation is the hardest additive deviation}]
\label{lem:kdev}
Suppose a center simultaneously exploits $k\in\{1,\ldots,m-1\}$ constituents, gaining $kd$ per locked-in round, and each targeted constituent independently exits at the next governance opportunity with probability $\alpha$. Under Assumption~\ref{ass:scale}, if the one-constituent accountability condition~\eqref{eq:IC-general} holds, then every such additive $k$-constituent deviation is also deterred.
\end{lemma}

\begin{lemma}[\textbf{The beginning of an epoch is the hardest time to deter exploitation}]
\label{lem:timing}
Suppose exploitation begins when $s\in\{1,\ldots,L\}$ locked-in rounds remain. If the accountability condition holds for $s=L$, it holds for every later point in the epoch with $s<L$.
\end{lemma}

\subsection{How Institutional Parameters Shape Delegation}
The general formulation shows that participation is governed by average cost savings whereas accountability is governed by the marginal retention value of one constituent. We now characterize how the two institutional thresholds move with the parameters and give closed-form comparative statics for the fixed-cost specialization. Table~\ref{tab:predictions} summarizes the main results.

\begin{theorem}[\textbf{Comparative statics under fixed costs}]
\label{thm:comparative}
Assume $\kappa>0$, $d>0$, $0<\delta<1$, $L\ge1$, $0<\alpha\le1$, and $0\le\ell<\kappa$. Let
\[
\munder=\frac{\kappa}{\kappa-\ell}
\quad\text{and}\quad
\mbar=\frac{1+\sqrt{1+4C}}{2},
\qquad
C=\frac{\alpha\kappa\delta^L}{d(1-\delta^L)}.
\]
Then
\begin{align}
&\frac{\partial \munder}{\partial\kappa}\le0 \quad\text{(strict if $\ell>0$)}, &\frac{\partial \munder}{\partial\ell}>0,\\
&\frac{\partial \mbar}{\partial\kappa}>0, &\frac{\partial \mbar}{\partial d}<0,\\
&\frac{\partial \mbar}{\partial\alpha}>0, &\frac{\partial \mbar}{\partial\delta}>0,
\end{align}
and $\mbar(L+1)<\mbar(L)$. Consequently, after integer rounding, $\mmin$ is weakly decreasing in $\kappa$ and weakly increasing in $\ell$, while $\mmax$ is weakly increasing in $\kappa$, $\alpha$, and $\delta$ and weakly decreasing in $d$ and $L$.
\end{theorem}

\begin{corollary}[\textbf{Stronger fixed-cost economies widen the viable interval}]
\label{cor:kappa}
For fixed $\ell,d,\alpha,\delta,L$, increasing $\kappa$ weakly lowers $\mmin$ and weakly raises $\mmax$. Thus stronger fixed-cost economies expand the set of constituency sizes that are simultaneously individually attractive and accountable under exit.
\end{corollary}

\begin{corollary}[\textbf{More reliable exit supports accountability at larger scales}]
\label{cor:alpha}
For fixed $\kappa,d,\delta,L$, $\mmax$ is weakly increasing in $\alpha$. If $\alpha=0$, individual exit supplies no expected punishment and no nontrivial constituency is accountable through this mechanism alone.
\end{corollary}

\begin{corollary}[\textbf{Long lock-in weakens accountability}]
\label{cor:lockin}
For fixed $\kappa,d,\alpha,\delta$, $\mmax$ is weakly decreasing in governance-epoch length $L$. Moreover, for $0<\delta<1$, $\mbar\to1$ as $L\to\infty$.
\end{corollary}

\begin{table}[t]
\caption{Comparative statics for the fixed-cost participation and accountability thresholds.}
\label{tab:predictions}
\centering
\small
\begin{tabular}{@{}p{0.14\columnwidth}ccp{0.50\columnwidth}@{}}
\toprule
Increase & $m_{\min}$ & $m_{\max}$ & Predicted institutional effect \\
\midrule
$\kappa$ & $\downarrow$ & $\uparrow$ & Delegation viable over a wider size range \\
$\ell$ & $\uparrow$ & -- & Harder to induce delegation \\
$d$ & -- & $\downarrow$ & Stronger agency problem; smaller accountable centers \\
$\alpha$ & -- & $\uparrow$ & More reliable detection and exit strengthen accountability \\
$\delta$ & -- & $\uparrow$ & Future exit matters more \\
$L$ & -- & $\downarrow$ & Longer lock-in weakens exit discipline \\
\bottomrule
\end{tabular}
\end{table}

\subsection{The Emergence of Polycentric Governance}
The two thresholds $\mmin$ and $\mmax$ characterize constituency sizes that are simultaneously individually attractive and accountable. Importantly, the following result depends only on the existence of these two thresholds; it is therefore not specific to the fixed-cost specialization.

\begin{definition}[\textbf{Fully delegated accountable partition}]
A partition of $n$ agents into $q$ constituencies with integer sizes $m_1,\ldots,m_q$ is \emph{fully delegated and accountable} if every constituency is nontrivial and satisfies
\[
\mmin\le m_j\le\mmax \qquad\text{for all }j=1,\ldots,q.
\]
\end{definition}

\begin{theorem}[\textbf{Feasible number of accountable centers}]
\label{thm:polycentric}
Suppose $2\le\mmin\le\mmax$. A fully delegated accountable partition of $n$ agents into exactly $q$ constituencies exists if and only if
\begin{equation}
\boxed{\left\lceil\frac{n}{\mmax}\right\rceil\le q\le\left\lfloor\frac{n}{\mmin}\right\rfloor}
\label{eq:qinterval}
\end{equation}
Consequently, if the interval in~\eqref{eq:qinterval} is nonempty and $n>\mmax$, every fully delegated accountable governance structure is polycentric, with at least
\begin{equation}
q_{\min}=\left\lceil\frac{n}{\mmax}\right\rceil\ge2
\end{equation}
centers.
\end{theorem}

\begin{corollary}[\textbf{Primitive condition for necessary polycentricity under fixed costs}]
\label{cor:primitive-polycentricity}
Under the fixed-cost specialization, define
\[
\mmin=\max\left\{2,\left\lceil\frac{\kappa}{\kappa-\ell}\right\rceil\right\},
\qquad
\mmax=\left\lfloor\frac{1+\sqrt{1+4\frac{\alpha\kappa\delta^L}{d(1-\delta^L)}}}{2}\right\rfloor.
\]
If $\mmin\le\mmax<n$ and $\lceil n/\mmax\rceil\le\lfloor n/\mmin\rfloor$, then every fully delegated accountable governance structure is necessarily polycentric.
\end{corollary}

\paragraph{\textbf{Takeaway}} 

Theorem~\ref{thm:polycentric} is a feasibility and necessity result, not an equilibrium-selection claim. It identifies a region, under benchmark conditions, in which \emph{economies of scale rule out very small delegated institutions} while \emph{accountability rules out a grand center}. Whether decentralized strategic agents discover and sustain such
structures once formation and replacement evolve endogenously is the question
we examine experimentally.

\subsection{Structural Polycentricity vs. Behavioural Diversity} \label{sec:Structural_Behavioural_Diversity}

Our theory characterizes when accountable delegation requires authority to be distributed across multiple centers. Such \emph{structural polycentricity}, however, does not necessarily imply diversity in how those centers behave. We therefore distinguish between structural polycentricity, which describes the distribution of authority across centers, and \emph{behavioural diversity}, which describes whether distinct centers make different decisions when facing comparable strategic conditions. In simple terms, a system can be highly polycentric in its formal organization while remaining behaviourally concentrated if its centers respond similarly to the same situations. This distinction is particularly relevant when centers are controlled by LLMs. Recent work documents substantial intra-model repetition and inter-model homogeneity in LLM responses to open-ended prompts \cite{jiang2026artificial}. Although these findings concern generative behaviour rather than strategic governance, they raise the possibility that formally independent LLM centers may nevertheless respond similarly to the same incentives. A governance system may therefore be highly polycentric in its institutional structure while remaining behaviourally concentrated if its centers repeatedly make similar exploitation decisions or similar mistakes. Our theoretical results do not require behavioural heterogeneity across centers, but this distinction motivates our empirical examination of how language-model agents exercise authority under comparable strategic conditions, including whether different models exhibit similar exploitation decisions and reasoning patterns.

%%%%%%%%%%%%%%%%%%%%%%%%%%%%%%%%%%%%%%%%%%%%%%%%%%%%%%%%%%%%%%%%%%%%%%%%%%%%%%%%%%%%%%%%%%%
%%%%%%%%%%%%%%%%%%%%%%%%%%%%%%%%%%%%%%%%%%%%%%%%%%%%%%%%%%%%%%%%%%%%%%%%%%%%%%%%%%%%%%%%%%%

\section{Polycentricity in Environments of AI Agents}

In this section, we study endogenous governance formation among large language model (LLM) agents. Our objective is to study whether polycentricity emerges naturally via their interactions, and whether the experimental results echo those of the theory developed in Section~\ref{sec:theory-model}. We start by describing the design of our simulation.

\subsection{Simulation Design} \label{sec:environment}

\paragraph{\textbf{Reduced-form Environment}}

The simulations implement the reduced-form game of Section~\ref{sec:theory-model}, with perfect detection and compulsory exit following exploitation ($\alpha=1$). Agents can establish new centres, step down as centres, and voluntarily join or leave constituencies at governance transitions. At each transition, exploited delegates first leave their centres. Existing centres then decide whether to continue, remaining delegates decide whether to stay, and free agents may volunteer to become new centres. If multiple agents volunteer, one is selected uniformly at random. Once the set of centres is determined, free agents choose between individual control and delegation to an eligible centre. A departing delegate may volunteer during the same transition.

Agents know the governance rules and observe the public constituency structure and their own actions and outcomes, but not other agents' private exploitation outcomes or reasoning. Each decision is accompanied by a history of the agent's observations, organized by governance epoch. Model-generated explanations are recorded for analysis but are neither shared with other agents nor included in subsequent memory.

\paragraph{\textbf{Matched Initialization Experiment}}

We ran five matched pairs of simulations, each with $n=30$ GPT-6 Luna agents. Each pair compared two starting conditions. In the \emph{centralized initialization}, one randomly selected agent began as the centre of a group containing all 30 agents. In the \emph{decentralized initialization}, all agents began independently and could volunteer to become a centre. One volunteer was selected at random, and the other agents could choose whether to join. After initialization, both conditions followed the same governance rules described in the previous paragraph.

The two simulations in each pair shared the same random draws but used separate model calls and memories. Each pair used a different random seed. We set $(\kappa,\ell,d,\delta,L)=(1,3/5,2/5,99/100,5),$
% \[
% (\kappa,\ell,d,\delta,L)=(1,3/5,2/5,99/100,5),
% \]
which gives the theoretical benchmark interval $3\leq m\leq7$ for $\alpha=1$, according to Corollaries~\ref{prop:mmin} and~\ref{prop:mmax}. We use this interval to classify constituency sizes. Appendix~\ref{app:parameter-sensitivity} has a detailed parameter sweep.

If an agent leaves a centre (e.g., due to exploitation), it cannot immediately rejoin. Before each later affiliation decision, the option to rejoin becomes available with probability $1/2$ and remains available unless the agent rejoins and leaves again. Rejoining is always voluntary. Agents see which centres they can join, but do not know the reinstatement rule or its random outcomes. This mechanism captures membership renewal in a finite population of persistent agents. Without it, permanent restrictions on rejoining would gradually limit agents' choices and could cause fragmentation unrelated to their strategic decisions. In real organizational settings, populations evolve as new members arrive and former members reconsider their affiliations. Probabilistic reinstatement provides an approximation of such membership renewal without introducing explicit population turnover, allowing constituencies to recover, reorganize, and compete for members over repeated governance epochs.

Each simulation contains 20 complete epochs. All agents received the same rules and one of five decision prompts: centre continuation, delegate release, volunteering, affiliation, or exploitation. Each prompt included the agent's current observations, available actions, and private memory. We used GPT-6 Luna with medium reasoning effort and required agents to provide both a structured decision and a private explanation. All prompts are provided in Appendix~\ref{app:agent-interface}, and the source code can be found in the supplement.

\paragraph{\textbf{Outcome Measures and Point-wise Comparison}}

We call a constituency \emph{viable} when $\mmin\leq m\leq \mmax$,\footnote{This classification identifies constituency sizes satisfying the analytical bounds, rather than establishing that the centre behaves non-exploitatively.} and an epoch \emph{polycentric} when at
least two viable constituencies coexist.
For each operational round, we record each agent's constituency size and whether they were exploited. To measure how cooperation gains are distributed, we compare each agent's payoff with the independent-control payoff of $-\kappa$ and attribute the difference to centres or delegates according to their role.

For the pointwise comparison, Sol evaluated each post-initialization Luna decision state using the same rules, observations, available actions, and history, without seeing Luna's decision or reasoning. Sol's responses did not affect the Luna trajectories. We compared 4,551 volunteering decisions and 1,988 exploitation decisions across the two models. The five random seeds were independent, but decisions within each simulation were not; we therefore report the results descriptively. This comparison captures differences between the models in identical decision states, rather than outcomes from independently evolving Sol societies.

%%%%%%%%%%%%%%%%%%%%%%%%%%%%%%%%%%%%%%%%%%%%%%%%%%%%%%%%%%%%%%%%%%%%%%%%%%%%%%%%%%%%%%%%%%%
%%%%%%%%%%%%%%%%%%%%%%%%%%%%%%%%%%%%%%%%%%%%%%%%%%%%%%%%%%%%%%%%%%%%%%%%%%%%%%%%%%%%%%%%%%%

\section{Simulation Results} \label{sec:results}

\begin{enumerate}[leftmargin=*]
    \item Luna agents repeatedly formed polycentric arrangements under both initialization conditions, but \textbf{constituency membership remained fluid} as fee-sharing incentives encouraged concentration and exploitation induced departures.
    \item Although delegation generated organizational cost savings, \textbf{centres captured most of these gains}. This concentration was particularly pronounced under centralized initialization, where exploitation left delegates collectively worse off than under individual control.
    \item Matched-state comparisons show that Luna and Sol were similarly willing to volunteer for leadership but \textbf{differed substantially in their exploitation decisions} and in the considerations emphasized in their explanations.
\end{enumerate}

\begin{figure}[!t]
  \centering
  \includegraphics[width=\linewidth]{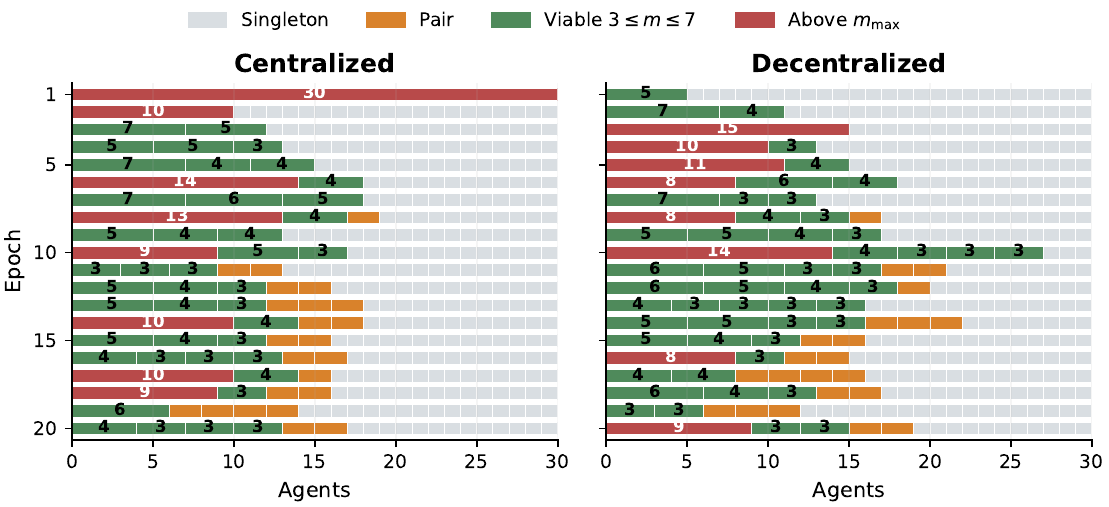}
  \Description{Six panels show complete constituency partitions over twenty epochs for three
  matched repeats. Centralized episodes occupy the first row and decentralized episodes the second.
  Colours distinguish singleton, size-two, viable, and oversized constituencies.}
  \caption{Constituency partitions at the start of each epoch. Left: centralized initialization, Right: decentralized initialization Each panel partitions all 30 agents; labels show constituency sizes of at least three. Multiple viable constituencies formed, dissolved, and re-formed under both initialization conditions. Appendix~\ref{app:all-partitions} shows all five repeats.}
  \label{fig:partitions}
\end{figure}

\subsection{Polycentricity Emerges but Remains Fluid}

Polycentricity emerged through both formation and fragmentation. All ten trajectories became polycentric, in all ten episodes and in 117 of the 200 observed epochs. This occurred under both centralized (55/100 epochs) and decentralized (62/100 epochs) initialization. Each trajectory was polycentric in 8--15 of its 20 epochs. The routes to polycentricity differed markedly for each of the two initialization conditions (Figure~\ref{fig:partitions}). Across both conditions, constituencies repeatedly expanded beyond the analytical accountability bound, lost members, and re-formed around existing or new centres. Thus,the agents did not simply select and maintain a partition within the theoretically compatible interval; polycentricity emerged through an process of concentration, exploitation, exit, and reconstruction.

Table~\ref{tab:initialization-results} compares the two initialization conditions across the full trajectories. Decentralized initialization produced slightly more polycentric epochs (62\% versus 55\%) and greater agent-round exposure to constituencies within the theoretical size interval (29.2\% versus 24.6\%). More importantly, exploitation exposure after the first epoch was lower (42.2\% versus 52.9\%). These descriptive differences suggest that the initial distribution of authority may influence subsequent governance dynamics, even though both conditions generated repeated formation and fragmentation of constituencies.

Exploitation often triggered the departure of targeted delegates, however these losses were not necessarily permanent: centres could attract new delegates, and departed agents could regain the opportunity to rejoin. Indeed, some centre explanations \emph{explicitly anticipated membership recovery when deciding whether to exploit}.

Every trajectory remained polycentric for at least two consecutive epochs, and one did so for ten. Some constituencies even grew to 14 agents, way above the theoretical bound. Nevertheless, larger constituencies subsequently fragmented, while smaller constituencies repeatedly lost members to competing centres. The observed dynamics were therefore characterized by recurring concentration, exploitation, exit, and re-formation, rather than convergence to a persistent governance structure.

\begin{table}[t]
\caption{Polycentricity and exploitation by initialization condition, aggregated over five matched pairs of 30-agent simulations (20 governance epochs per episode). An epoch is classified as polycentric when at least two constituencies have sizes within the theoretical benchmark interval $3\le m\le7$. Exploitation rates are calculated in epochs 2--20 to avoid the distinct initialization epoch.}
\label{tab:initialization-results}
\centering
\resizebox{\columnwidth}{!}{%
\begin{tabular}{lrr}
\toprule
Outcome & Centralized & Decentralized \\
\midrule
Polycentric epochs & $55/100$ & $62/100$ \\
Agent-rounds within theoretical bounds & $24.6\%$ & $29.2\%$ \\
Exploited delegate-rounds after epoch 1 & $52.9\%$ & $42.2\%$ \\
\bottomrule
\end{tabular}}
\end{table}

\paragraph{\textbf{Centres Captured Most of the Surplus}}

Although fee sharing creates aggregate gains relative to individual control, exploitation can redistribute these gains from delegates to centres. To examine who benefits from delegation, we measured each agent's payoff gain relative to independent control and aggregate these gains separately for centres and delegates.

For the centralized initialization, centres captured $104.4\%$ of aggregate delegation surplus, leaving delegates with a negative share of $-4.4\%$. A share exceeding $100\%$ is possible -- i.e., delegates can be collectively worse off than under independent control -- because exploitation transfers payoff from delegates to centres. Delegate surplus was negative in four of the five centralized repeats. For the decentralized initialization, centres captured $91.9\%$ of the surplus, while delegates retained $8.1\%$, with positive delegate surplus in all five repeats. Thus, although centres received most of the organizational gains under both conditions, the distribution of these gains differed based on initial governance structures (see Appendix~\ref{app:delegation-surplus}).

\subsection{Exercise of Authority}

We next investigate whether agents' governance decisions differ across models of varying capabilities. To this end, we compare a higher-capacity model (GPT-6 Sol) with a smaller model (GPT-6 Luna), evaluating their interest in assuming authority and their exploitation decisions under identical strategic conditions.

\paragraph{\textbf{The models were similarly willing to become centers, but differed substantially in how much they exploited.}} 

In 4,551 matched volunteering states, Luna and Sol volunteered in 96.53\% and 99.17\% of cases, respectively, agreeing in 96.04\% of states. The difference was substantially larger for exploitation. In 1,988 matched centre-decision states, Sol initiated new exploitation\footnote{I.e., initiating exploitation of at least one additional delegate who was not previously being exploited during the current governance epoch.} in 41.80\% of cases, compared with 31.79\% for Luna, \emph{a 10 percentage-points difference}. The models agreed on whether to exploit in 65.74\% of these states. This difference was particularly pronounced in the decentralized-initialization trajectories, where Sol exploited in $44.33\%$ of states compared with $28.11\%$ for Luna (a $16.23$ percentage-point difference). In contrast, their exploitation rates were much closer in the centralized-initialization trajectories ($38.64\%$ versus $36.38\%$, respectively). These results suggest that willingness to assume a leadership role provides little information about how an agent will exercise authority once that role has been acquired.

Figure~\ref{fig:exploitation-models} compares the models' exploitation decisions across constituency sizes. Within the theoretical benchmark interval ($3\le m\le7$), Sol initiated exploitation in $31.46\%$ of matched states, compared with $20.76\%$ for Luna. Above the accountability threshold ($m>7$), these rates increased to $44.00\%$ and $28.00\%$, respectively. Moreover, Sol continued to target a similar proportion of delegates in larger constituencies, whereas Luna targeted a smaller proportion (see also Appendix~\ref{app:exploitation-intensity}).

Furthermore, the models differed not simply in their willingness to exploit, but in \emph{how they assessed its longer-term consequences}. To investigate these differences, we analysed the transcripts using Non-negative Matrix Factorisation (NMF)~\cite{lee1999learning} to identify four recurring themes in the models' written explanations (see Appendix~\ref{app:reasoning-topics}). In the 440 matched states where Sol exploited but Luna did not, Luna placed much greater emphasis on the future fee increases caused by losing delegates (the mean normalized weight assigned to long-term fee loss was $67.3\%$ for Luna and $3.3\%$ for Sol). Conversely, Sol emphasized the possibility that departing constituents would soon be replaced, either by returning members or by new delegates, allowing the centre to maintain its constituency size (64.5\% versus 15.0\% for Luna). These differences relate to accountability: departure disciplines a centre only when the resulting membership loss is sufficiently durable. Luna's explanations primarily represented departure as a future fee cost, whereas Sol's more often represented it as turnover that could be offset by retained or replacement members. Nevertheless, note that these findings describe differences in the models' generated explanations, rather than establishing what caused their decisions.

\begin{figure}[t]
  \centering
  \includegraphics[width=0.7\columnwidth]{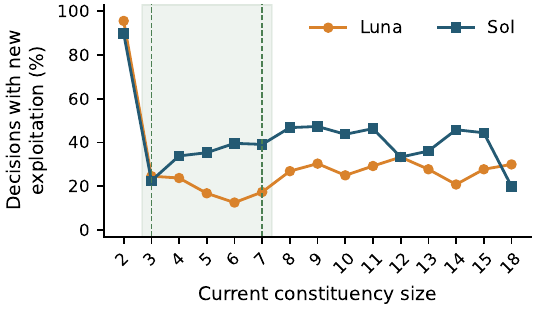}
  \Description{Two line charts compare Luna and Sol exploitation incidence and selected delegate
  shares across constituency sizes.}
  \caption{Post-initialization matched-state comparison of exploitation decisions across constituency sizes. The shaded area corresponds to the theoretical accountable interval $3\le m\le7$. Sol generally initiated exploitation more frequently as compared to Luna.}
  \label{fig:exploitation-models}
\end{figure}

\paragraph{\textbf{Exploitation Scope}}

Sol exploited more frequently, while Luna targeted more delegates when it did exploit. Conditional on initiating exploitation, Luna targeted $81.57\%$ of eligible delegates, compared with $74.25\%$ for Sol. However, because Sol initiated exploitation more frequently, it selected more new targets per decision on average ($1.011$ versus $0.670$). The models also differed in timing (Figure~\ref{fig:exploitation-timing}). Their exploitation rates were similar in the first operational round, but, counter-intuitively, Sol was more likely to initiate exploitation of additional delegates in subsequent rounds, particularly rounds two to four. This is in contrast to Lemma~\ref{lem:timing}, which establishes that exploitation is more beneficial to the center at the beginning of an epoch, since the centre can benefit from it for the greatest number of remaining rounds.

These findings also illustrate the distinction between structural polycentricity and behavioural diversity discussed in Section~\ref{sec:Structural_Behavioural_Diversity}. Even when centres face identical strategic conditions, the model controlling them can influence their exploitation decisions and their assessment of future membership losses. Thus, distributing decision-making power across multiple centres does not, by itself, determine how those centres will behave. Our matched-state comparison demonstrates behavioural differences across model families, although it does not establish whether centres controlled by the same model make correlated decisions.

\begin{figure}[t]
  \centering
  \includegraphics[width=0.8\columnwidth]{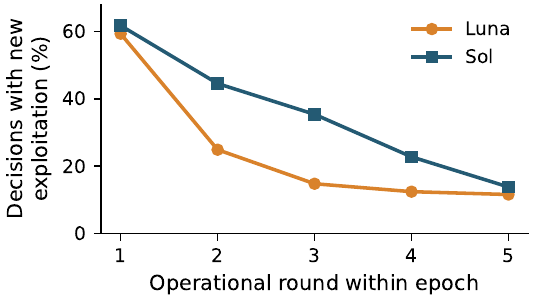}
  \Description{A line chart compares the share of Luna and Sol decisions that initiated new
  exploitation in each of the five operational rounds. The rates are similar in round one, after
  which Sol remains substantially higher through round four.}
  \caption{Timing of new exploitation within a governance epoch. New exploitation refers to targeting at least one delegate who was not previously being exploited during that epoch. Luna and Sol behaved similarly in the first operational round, but Sol initiated new exploitation more frequently in subsequent rounds, particularly rounds two to four.}
  \label{fig:exploitation-timing}
\end{figure}

%%%%%%%%%%%%%%%%%%%%%%%%%%%%%%%%%%%%%%%%%%%%%%%%%%%%%%%%%%%%%%%%%%%%%%%%
%%%%%%%%%%%%%%%%%%%%%%%%%%%%%%%%%%%%%%%%%%%%%%%%%%%%%%%%%%%%%%%%%%%%%%%%

\section{Discussion}

Our theoretical and empirical results highlight a central \emph{tension} in delegated governance: the economies of scale that encourage agents to form larger constituencies can also weaken the accountability of their centres. The theory identifies a minimum size at which delegation becomes individually attractive and a maximum size at which the threat of individual exit can deter exploitation. However, this is a feasibility result, not a prediction that agents will form or maintain such a partition. In fact, in the simulations, while polycentric arrangements emerged under both initialization conditions, constituencies repeatedly grew, fragmented, and re-formed.

A key insight concerns the \emph{robustness of exit as an accountability mechanism}. 
In the simulations, exploited delegates always leave at the next governance transition, but centres may subsequently attract replacement members. Thus, even certain departure does not necessarily provide a strong deterrent when the resulting membership loss may be temporary. This helps explain why exploitation can persist in real life despite compulsory exit and why constituencies may grow beyond the theoretical accountability threshold.

Our results also distinguish the \emph{efficiency of delegation} from the \emph{distribution of its benefits}. Sharing organizational costs generates aggregate gains, but exploitation allows centres to appropriate those gains at their delegates' expense. Indeed, following centralized initialization, delegates were collectively worse off than under independent control, despite positive aggregate delegation surplus. Delegates fared better under decentralized initialization, although centres still captured most of the gains. These findings suggest that the emergence of collective institutions, or even multiple competing institutions, is not by itself sufficient evidence that delegation benefits those who surrender control.

Finally, the comparison between Luna and Sol illustrates one element of why structural polycentricity should be distinguished from behavioural diversity. The models were similarly willing to become centres but differed substantially in their exploitation decisions, even when presented with identical strategic states. The behaviour of the models exercising authority matters alongside the institutional structure in which they operate.

\begin{figure}[t]
  \centering
  \includegraphics[width=0.9\linewidth]{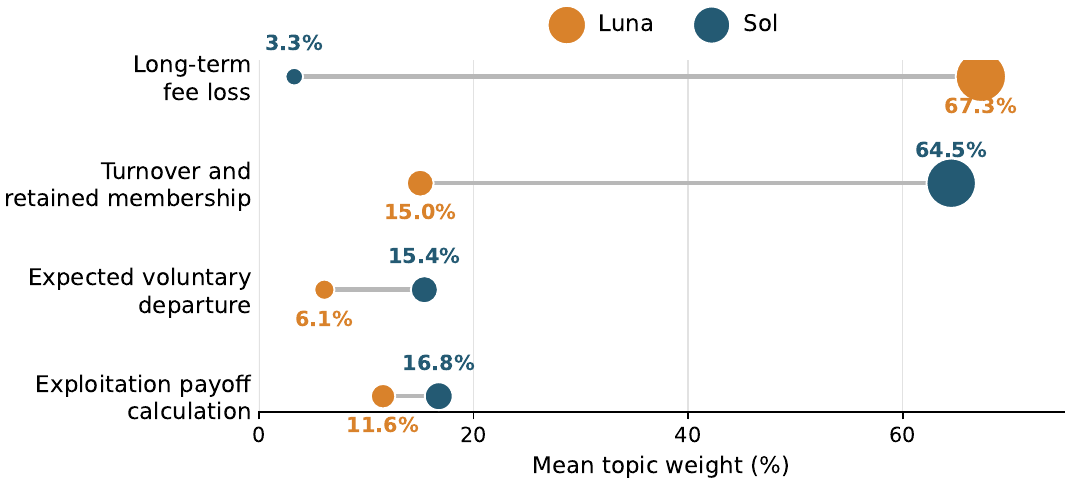}
  \Description{Comparison of reason in the cases
  where Sol exploited and Luna did not. Luna strongly emphasizes long-term fee loss, while Sol
  strongly emphasizes turnover and retained membership.}
  \caption{Themes in model-generated explanations for matched states where Sol initiated exploitation but Luna did not. Position and marker size represent each model's mean topic weight. Luna emphasized long-term fee losses, whereas Sol emphasized membership turnover and retained constituency size.}
  \label{fig:reason-topics}
\end{figure}

%%%%%%%%%%%%%%%%%%%%%%%%%%%%%%%%%%%%%%%%%%%%%%%%%%%%%%%%%%%%%%%%%%%%%%%%
%%%%%%%%%%%%%%%%%%%%%%%%%%%%%%%%%%%%%%%%%%%%%%%%%%%%%%%%%%%%%%%%%%%%%%%%

\section{Conclusion}

We studied how economies of scale and exit-based accountability shape voluntary delegation to self-interested centres. Our theory identifies the constituency sizes at which delegation is attractive and exploitation can be deterred, and characterizes when fully delegated accountable governance must be polycentric. Simulations with LLM agents show that polycentric arrangements repeatedly emerge, but remain fluid as constituencies expand, exploit, lose members, and re-form. They also reveal that centres can capture most of the gains from delegation, and that different language models may exercise identical authority in substantially different ways.

Together, these findings suggest that distributing authority across multiple centres is not sufficient to ensure accountable governance. What matters is whether exit imposes meaningful costs on those exercising authority, and how those agents respond to these incentives. Beyond these findings, our reduced-form environment can serve as a benchmark for the multi-agent systems community, enabling researchers to systematically compare AI agents and governance mechanisms under different incentives for delegation, exploitation, and exit.

%%%%%%%%%%%%%%%%%%%%%%%%%%%%%%%%%%%%%%%%%%%%%%%%%%%%%%%%%%%%%%%%%%%%%%%%

%%% The acknowledgments section is defined using the "acks" environment
%%% (rather than an unnumbered section). The use of this environment 
%%% ensures the proper identification of the section in the article 
%%% metadata as well as the consistent spelling of the heading.

% \begin{acks}
 
% \end{acks}

%%%%%%%%%%%%%%%%%%%%%%%%%%%%%%%%%%%%%%%%%%%%%%%%%%%%%%%%%%%%%%%%%%%%%%%%

%%% The next two lines define, first, the bibliography style to be 
%%% applied, and, second, the bibliography file to be used.

\bibliographystyle{ACM-Reference-Format} 
\bibliography{bibliography}

%%%%%%%%%%%%%%%%%%%%%%%%%%%%%%%%%%%%%%%%%%%%%%%%%%%%%%%%%%%%%%%%%%%%%%%%

\clearpage

\appendix

\section{Proofs} \label{app: Proofs}

\begin{proof}[Proof of~\Cref{prop:general-mmin}] \textbf{(General participation threshold)}\\
Because $c(m)$ is weakly decreasing, $B(m)=c(1)-c(m)-\ell$ is weakly increasing in $m$. Hence, once $B(m)\ge0$ at some size, the inequality holds at every larger size. The smallest nontrivial size satisfying the inequality is therefore exactly~\eqref{eq:mmin-general}.
\end{proof}

\begin{proof}[Proof of~\Cref{prop:mmin}] \textbf{(Minimum viable constituency size under fixed costs)}\\
Under~\eqref{eq:fixed-cost}, condition~\eqref{eq:participation-benefit} becomes $\kappa(1-1/m)\ge\ell$. Rearranging gives $m\ge\kappa/(\kappa-\ell)$ when $\ell<\kappa$. Taking the smallest nontrivial integer gives~\eqref{eq:mmin}. If $\ell\ge\kappa$, then $\kappa(1-1/m)<\kappa\le\ell$ for every finite $m$.
\end{proof}

\begin{proof}[Proof of~\Cref{prop:general-accountability}]
\textbf{(General accountability condition)}
If the center exploits at the beginning of the epoch, its expected deviation gain is the finite geometric sum
\begin{align}
V_{\mathrm{dev}}
&=d+\delta d+\cdots+\delta^{L-1}d\\
&=d\frac{1-\delta^L}{1-\delta}
\label{eq:devgain}
\end{align}

If the delegate exits, the center's per-round organizational cost rises from $c(m)$ to $c(m-1)$, a loss of $D(m)$. Exit occurs with probability $\alpha$ and can affect payoffs only after the interaction survives the $L$ locked-in rounds, which contributes the factor $\delta^L$. Conditional on exit, the loss persists, so the loss from future exit is
\begin{align}
V_{\mathrm{exit}}
&=\alpha \left( \delta^L D(m)+\delta^{L+1}D(m)+\cdots \right)\\
&=\alpha\frac{\delta^L D(m)}{1-\delta}
\label{eq:exitloss}
\end{align}

Non-exploitation is incentive compatible exactly when $V_{\mathrm{exit}}\ge V_{\mathrm{dev}}$. Cancelling the common positive factor $1/(1-\delta)$ yields~\eqref{eq:IC-general}. Since $D(m)$ is weakly decreasing, if~\eqref{eq:IC-general} fails at some size it also fails at every larger size, so accountable sizes are downward closed.
\end{proof}

\begin{proof}[Proof of~\Cref{prop:mmax}]
\textbf{(Maximum accountable constituency size under fixed costs)}
Under $c(m)=\kappa/m$,
\[
D(m)=\frac{\kappa}{m-1}-\frac{\kappa}{m}=\frac{\kappa}{m(m-1)}.
\]
Substituting this expression into~\eqref{eq:IC-general} gives~\eqref{eq:IC}--\eqref{eq:ICquadratic}. Let $C=\alpha\kappa\delta^L/[d(1-\delta^L)]$. The inequality is $m^2-m-C\le0$. Its positive root is $(1+\sqrt{1+4C})/2$, and taking the greatest integer not exceeding that root gives~\eqref{eq:mmax}.
\end{proof}

\begin{proof}[Proof of~\Cref{lem:kdev}]
\textbf{(Single-constituent exploitation is the hardest additive deviation)}
Let $X_k$ be the number of the $k$ targeted constituents who exit, so $X_k\sim\mathrm{Binomial}(k,\alpha)$. The center's expected persistent per-round cost increase after the governance boundary is
\[
\overline D_{m,k}=\mathbb{E}[c(m-X_k)-c(m)].
\]
Construct $X_k$ from independent Bernoulli variables and add targets one at a time. If $X_{r-1}$ departures have resulted from the first $r-1$ targets, the $r^{\text{th}}$ target increases expected future cost by
\[
\alpha\,\mathbb{E}\left[D(m-X_{r-1})\right]\ge \alpha D(m),
\]
where the inequality follows because $m-X_{r-1}\le m$ and $D(\cdot)$ is weakly decreasing in constituency size. Summing over $r=1,\ldots,k$ gives $\overline D_{m,k}\ge k\alpha D(m)$. The expected deviation gain is $kd(1-\delta^L)/(1-\delta)$, while the expected future loss is $\delta^L\overline D_{m,k}/(1-\delta)$. Thus the one-target condition $\alpha\delta^L D(m)\ge d(1-\delta^L)$ implies the $k$-target condition.
\end{proof}

\begin{proof}[Proof of~\Cref{lem:timing}]
\textbf{(The beginning of an epoch is the hardest time to deter exploitation)}
For fixed $m$, the accountability condition with $s$ rounds remaining is
\[
\frac{\alpha D(m)}{d}\frac{\delta^s}{1-\delta^s}\ge1.
\]
For $0<\delta<1$, $\delta^s$ decreases as $s$ increases, while $z/(1-z)$ is increasing in $z$. Hence $\delta^s/(1-\delta^s)$ is decreasing in $s$, making the condition hardest to satisfy at the largest possible $s$, namely $s=L$.
\end{proof}

\begin{proof}[Proof of~\Cref{thm:comparative}]
\textbf{(Comparative statics under fixed costs)}
For the participation threshold, under the fixed-cost specialization
$c(m)=\kappa/m$, the continuous threshold is
\[
\munder=\frac{\kappa}{\kappa-\ell},
\]
defined for $\kappa>\ell$. Direct differentiation with respect to $\kappa$ gives
\[
\frac{\partial\munder}{\partial\kappa}
=
\frac{(\kappa-\ell)-\kappa}{(\kappa-\ell)^2}
=
-\frac{\ell}{(\kappa-\ell)^2}
\le 0,
\]
with strict inequality whenever $\ell>0$. Thus a larger shared organizational
cost makes delegation attractive at a weakly smaller constituency size,
because the gains from sharing that cost become larger.

Similarly,
\[
\frac{\partial\munder}{\partial\ell}
=
\frac{\kappa}{(\kappa-\ell)^2}
>0.
\]
Hence a larger private cost of delegation raises the minimum constituency size
required for delegation to be individually attractive.

For the accountability threshold, define
\[
C
=
\frac{\alpha\kappa\delta^L}
{d(1-\delta^L)}
\]
and
\[
g(C)
=
\frac{1+\sqrt{1+4C}}{2}.
\]
Then the continuous accountability threshold is
\[
\mbar=g(C),
\]
and
\[
g'(C)
=
\frac{1}{\sqrt{1+4C}}
>0.
\]
Therefore, the sign of any comparative static of $\mbar$ is the same as the
sign of the corresponding comparative static of $C$.

First, with respect to $\kappa$,
\[
\frac{\partial C}{\partial\kappa}
=
\frac{\alpha\delta^L}
{d(1-\delta^L)}
\ge 0.
\]
If $\alpha>0$, then
\[
\frac{\partial C}{\partial\kappa}>0,
\]
and therefore
\[
\frac{\partial\mbar}{\partial\kappa}
=
g'(C)\frac{\partial C}{\partial\kappa}
>0.
\]
Thus, when exploitation can lead to exit, a larger shared organizational cost
raises the value to the center of retaining a constituent and therefore permits
accountability at a larger constituency size. If $\alpha=0$, exit never follows
exploitation, $C=0$, and changing $\kappa$ has no effect on exit-based
accountability.

Second, with respect to the exploitation gain $d$,
\[
\frac{\partial C}{\partial d}
=
-\frac{\alpha\kappa\delta^L}
{d^2(1-\delta^L)}
\le 0.
\]
For $\alpha>0$,
\[
\frac{\partial C}{\partial d}<0,
\]
so
\[
\frac{\partial\mbar}{\partial d}
=
g'(C)\frac{\partial C}{\partial d}
<0.
\]
A larger private gain from exploitation therefore lowers the maximum
constituency size at which exit can discipline the center.

Third, consider the probability $\alpha$ that exploitation is detected and
results in the targeted constituent's departure at the next governance
opportunity. We have
\[
\frac{\partial C}{\partial\alpha}
=
\frac{\kappa\delta^L}
{d(1-\delta^L)}
>0.
\]
Hence
\[
\frac{\partial\mbar}{\partial\alpha}
=
g'(C)\frac{\partial C}{\partial\alpha}
>0.
\]
A higher probability that exploitation results in exit strengthens the
expected punishment from exploiting a constituent and therefore increases the
maximum accountable constituency size.

For the continuation probability $\delta$, write
\[
z=\delta^L.
\]
Then
\[
C
=
\frac{\alpha\kappa}{d}\frac{z}{1-z}.
\]
Since
\[
\frac{d}{dz}\frac{z}{1-z}
=
\frac{1}{(1-z)^2}
>0
\]
and
\[
\frac{\partial z}{\partial\delta}
=
L\delta^{L-1}
>0,
\]
the chain rule gives
\[
\frac{\partial C}{\partial\delta}
=
\frac{\alpha\kappa}{d}
\frac{L\delta^{L-1}}
{(1-\delta^L)^2}
\ge 0.
\]
For $\alpha>0$ this inequality is strict, and hence
\[
\frac{\partial\mbar}{\partial\delta}
=
g'(C)\frac{\partial C}{\partial\delta}
>0.
\]
Thus greater patience strengthens accountability because the future loss from a
constituent's departure receives greater weight relative to the current-epoch
gain from exploitation.

Finally, $L$ is integer-valued. Because $0<\delta<1$,
\[
\delta^{L+1}<\delta^L.
\]
The function
\[
h(z)=\frac{z}{1-z}
\]
is strictly increasing on $(0,1)$ because
\[
h'(z)=\frac{1}{(1-z)^2}>0.
\]
It follows that
\[
\frac{\delta^{L+1}}{1-\delta^{L+1}}
<
\frac{\delta^L}{1-\delta^L}.
\]
Multiplying both sides by the nonnegative factor $\alpha\kappa/d$ gives
\[
C(L+1)\le C(L),
\]
with strict inequality whenever $\alpha>0$. Since $g$ is strictly increasing,
\[
\mbar(L+1)\le\mbar(L),
\]
again with strict inequality for $\alpha>0$. Hence longer governance lock-in
weakens exit-based accountability because exploitation can continue for longer
before the affected constituent has an opportunity to leave.

The preceding statements concern the continuous thresholds $\munder$ and
$\mbar$. The integer thresholds are obtained as
\[
\mmin=\max\left\{2, \left\lceil\munder\right\rceil \right\}
\qquad\text{and}\qquad
\mmax=\left\lfloor\mbar\right\rfloor.
\]
Therefore, the integer thresholds inherit the corresponding weak
monotonicities, although changes in a primitive need not change the integer
threshold when the continuous threshold remains within the same interval
between consecutive integers.
\end{proof}

\begin{proof}[Proof of~\Cref{cor:lockin}]
\textbf{(Long lock-in weakens accountability)}
Monotonicity follows from Theorem~\ref{thm:comparative}. As $L\to\infty$, $\delta^L\to0$, so $C\to0$ and~\eqref{eq:mbar} gives $\mbar\to1$.
\end{proof}

\begin{proof}[Proof of~\Cref{thm:polycentric}]
\textbf{(Feasible number of accountable centers)}
Assume $\mmin\le\mmax$.

\emph{(Only if)}
Consider any fully delegated accountable partition consisting of $q$
constituencies, with constituency sizes
\[
m_1,\ldots,m_q.
\]
Because the partition is fully delegated and accountable, every constituency
must satisfy
\[
\mmin\le m_j\le\mmax
\qquad\text{for all }j=1,\ldots,q,
\]
and the constituencies partition the population, so
\[
\sum_{j=1}^q m_j=n.
\]

Since every constituency has size at most $\mmax$,
\[
n
=
\sum_{j=1}^q m_j
\le
\sum_{j=1}^q \mmax
=
q\mmax.
\]
Therefore
\[
q\ge\frac{n}{\mmax}.
\]
Because $q$ is an integer, this implies
\[
q\ge
\left\lceil
\frac{n}{\mmax}
\right\rceil.
\]

Likewise, since every constituency has size at least $\mmin$,
\[
n
=
\sum_{j=1}^q m_j
\ge
\sum_{j=1}^q \mmin
=
q\mmin.
\]
Hence
\[
q\le\frac{n}{\mmin}.
\]
Again using that $q$ is integer,
\[
q\le
\left\lfloor
\frac{n}{\mmin}
\right\rfloor.
\]

Thus any fully delegated accountable partition with $q$ centers must satisfy
\[
\left\lceil
\frac{n}{\mmax}
\right\rceil
\le q\le
\left\lfloor
\frac{n}{\mmin}
\right\rfloor,
\]
so~\eqref{eq:qinterval} is necessary.

\emph{(If)}
Now take any integer $q$ satisfying~\eqref{eq:qinterval}. We construct a fully
delegated accountable partition with exactly $q$ constituencies.

Begin by assigning $\mmin$ agents to each of the $q$ constituencies. This uses
\[
q\mmin
\]
agents. Because
\[
q
\le
\left\lfloor
\frac{n}{\mmin}
\right\rfloor,
\]
we have
\[
q\mmin\le n.
\]
Hence the number of agents remaining after this initial assignment is
\[
r=n-q\mmin\ge0.
\]

Each constituency can receive at most
\[
\mmax-\mmin
\]
additional agents before reaching the maximum accountable size $\mmax$.
Therefore the total remaining capacity across the $q$ constituencies is
\[
q(\mmax-\mmin).
\]

Since
\[
q
\ge
\left\lceil
\frac{n}{\mmax}
\right\rceil,
\]
we have
\[
q\mmax\ge n.
\]
Subtracting $q\mmin$ from both sides gives
\[
q(\mmax-\mmin)
\ge
n-q\mmin
=
r.
\]
Thus the total remaining capacity is at least as large as the number of
unassigned agents.

We may therefore distribute the $r$ remaining agents among the $q$
constituencies, adding agents one at a time to any constituency whose size is
still below $\mmax$. Because the aggregate remaining capacity is at least $r$,
all agents can be assigned before any constituency exceeds $\mmax$.

The resulting constituency sizes satisfy
\[
\mmin\le m_j\le\mmax
\qquad\text{for all }j,
\]
and
\[
\sum_{j=1}^q m_j=n.
\]
Hence we have constructed a fully delegated accountable partition with exactly
$q$ centers, proving sufficiency.

Finally, suppose
\[
n>\mmax.
\]
Then
\[
\frac{n}{\mmax}>1,
\]
and therefore
\[
\left\lceil
\frac{n}{\mmax}
\right\rceil\ge2.
\]
Since every fully delegated accountable partition must contain at least
$\lceil n/\mmax\rceil$ centers, every such partition must contain at least two
centers. Thus, whenever the population exceeds the maximum accountable
constituency size, full accountable delegation necessarily requires
polycentric governance.
\end{proof}

\section{Statement on AI Use}

The code was written using assistance from Codex. Codex was mainly used to make the code more structured, efficient and for implementing specific design patterns. The code specification was done manually and all code outputs were verified to check that the implementation matched the specification.

\section{Experimental Details}
\label{app:experimental-details}

Table~\ref{tab:experimental-parameters} reports the fixed simulation parameters. Main trajectories
used GPT-6 Luna with medium reasoning effort and low text verbosity. Exact pointwise comparisons
replayed observed volunteering and exploitation decisions with GPT-6 Sol under the same model
settings, prompts, observations, legal actions, and private memories. Sol was not shown Luna's
action or rationale.

\begin{table}[h]
\caption{Simulation parameters.}
\label{tab:experimental-parameters}
\centering
\small
\begin{tabular}{@{}p{0.24\columnwidth}p{0.18\columnwidth}p{0.47\columnwidth}@{}}
\toprule
Quantity & Value & Meaning \\
\midrule
$n$ & 30 & Persistent agents \\
$\kappa$ & 1 & Constituency fee per round \\
$\ell$ & $3/5$ & Delegation cost per round \\
$d$ & $2/5$ & Per-target transfer to the center \\
$\delta$ & $99/100$ & Continuation probability after each round \\
$L$ & 5 & Operational rounds per epoch \\
$T$ & 100 & Right-censoring cap in rounds \\
Detection and exit & Perfect & Exploited delegates leave at the governance boundary \\
Reinstatement & $1/2$ & Hidden draw per prohibited agent--center relationship \\
Repeats & 5 matched pairs & Independent exogenous tapes between pairs \\
\bottomrule
\end{tabular}
\end{table}

For these values,
\[
m_{\min}=\max\!\left\{2,
\left\lceil\frac{\kappa}{\kappa-\ell}\right\rceil\right\}=3,
\qquad
m_{\max}=\left\lfloor
\frac{1+\sqrt{1+4\frac{\kappa\delta^L}{d(1-\delta^L)}}}{2}
\right\rfloor=7.
\]

\subsection{Parameter sensitivity}
\label{app:parameter-sensitivity}

Under fixed costs, the thresholds are
\[
m_{\min}=\max\!\left\{2,
\left\lceil\frac{\kappa}{\kappa-\ell}\right\rceil\right\},
\qquad
m_{\max}=\left\lfloor
\frac{1+\sqrt{1+4\frac{\alpha\kappa\delta^L}{d(1-\delta^L)}}}{2}
\right\rfloor.
\]

\begin{table}[t]
\caption{Viable interval $[m_{\min},m_{\max}]$ as $\ell$ and $d$ vary, with
$\kappa=1$, $\alpha=1$, $\delta=0.99$, and $L=5$. A dash denotes an empty interval; bold marks the
simulated setting.}
\label{tab:ell-d-sensitivity}
\centering
\footnotesize
\setlength{\tabcolsep}{2.5pt}
\begin{tabular}{c|ccccc}
\toprule
$\ell\backslash d$ & 0.2 & 0.4 & 0.6 & 0.8 & 1.0 \\
\midrule
0.2 & $[2,10]$ & $[2,7]$ & $[2,6]$ & $[2,5]$ & $[2,4]$ \\
0.4 & $[2,10]$ & $[2,7]$ & $[2,6]$ & $[2,5]$ & $[2,4]$ \\
0.6 & $[3,10]$ & $\mathbf{[3,7]}$ & $[3,6]$ & $[3,5]$ & $[3,4]$ \\
0.8 & $[5,10]$ & $[5,7]$ & $[5,6]$ & $[5,5]$ & -- \\
0.9 & $[10,10]$ & -- & -- & -- & -- \\
\bottomrule
\end{tabular}
\end{table}

Increasing $\ell$ raises $m_{\min}$, while increasing $d$ lowers $m_{\max}$.

\begin{table}[t]
\caption{$m_{\max}$ as $\delta$ and $L$ vary, with $\kappa=1$, $d=0.4$, and $\alpha=1$.
Bold marks the simulated setting.}
\label{tab:delta-L-sensitivity}
\centering
\footnotesize
\setlength{\tabcolsep}{5pt}
\begin{tabular}{c|rrrr}
\toprule
$\delta\backslash L$ & 1 & 5 & 10 & 20 \\
\midrule
0.900 & 5 & 2 & 1 & 1 \\
0.950 & 7 & 3 & 2 & 1 \\
0.980 & 11 & 5 & 3 & 2 \\
0.990 & 16 & \textbf{7} & 5 & 3 \\
0.995 & 22 & 10 & 7 & 5 \\
\bottomrule
\end{tabular}
\end{table}

Increasing $\delta$ raises $m_{\max}$, while increasing $L$ lowers it.

\begin{table}[t]
\caption{$m_{\max}$ as the accountability probability $\alpha$ varies, with $\kappa=1$,
$d=0.4$, $\delta=0.99$, and $L=5$. Bold marks the simulated setting.}
\label{tab:alpha-sensitivity}
\centering
\footnotesize
\setlength{\tabcolsep}{7pt}
\begin{tabular}{c|rrrrr}
\toprule
$\alpha$ & 0.00 & 0.25 & 0.50 & 0.75 & 1.00 \\
\midrule
$m_{\max}$ & 1 & 4 & 5 & 6 & \textbf{7} \\
\bottomrule
\end{tabular}
\end{table}

More reliable detection and exit raise the maximum accountable constituency size. At $\alpha=0$,
no nontrivial constituency is accountable through exit alone.

\begin{table}[t]
\caption{Thresholds as $\kappa$ varies, with $\ell=0.6$, $d=0.4$, $\alpha=1$,
$\delta=0.99$, and $L=5$. A dash denotes that no finite $m_{\min}$ exists; bold marks the simulated
setting.}
\label{tab:kappa-sensitivity}
\centering
\footnotesize
\setlength{\tabcolsep}{7pt}
\begin{tabular}{c|cc}
\toprule
$\kappa$ & $m_{\min}$ & $m_{\max}$ \\
\midrule
0.50 & -- & 5 \\
0.75 & 5 & 6 \\
\textbf{1.00} & \textbf{3} & \textbf{7} \\
1.25 & 2 & 8 \\
1.50 & 2 & 9 \\
\bottomrule
\end{tabular}
\end{table}

Increasing $\kappa$ lowers $m_{\min}$ and raises $m_{\max}$. If $\kappa\leq\ell$, no finite
constituency makes delegation individually rational through fixed-cost sharing. Together, the four
sweeps reproduce the comparative directions in Table~\ref{tab:predictions}.

\subsection{Reasoning-topic analysis}
\label{app:reasoning-topics}

We analyzed the 3,976 explanations from all 1,988 post-initialization exploitation-state pairs.
Agent identifiers and numeric magnitudes were masked before TF--IDF vectorization with unigrams and
bigrams. We fitted a four-component non-negative matrix factorization with seed 2027; model identity
and action were not inputs. We named each component from its highest-weight terms and inspected the
three matched pairs with the largest absolute model difference for that component. Topic weights
were normalized to sum to one within each explanation.

\begin{table*}[t]
\caption{Reasoning-topic weights by matched action outcome. The middle columns report Sol minus
Luna in percentage points; positive values indicate greater Sol emphasis. The final columns give
the models' absolute mean weights in the states where only Sol exploited.}
\label{tab:reasoning-topics}
\centering
\footnotesize
\setlength{\tabcolsep}{4.2pt}
\begin{tabular}{p{3.4cm}p{5.2cm}rrrrrr}
\toprule
& & \multicolumn{4}{c}{Sol minus Luna (percentage points)}
  & \multicolumn{2}{c}{Sol-only states (\%)} \\
\cmidrule(lr){3-6}\cmidrule(lr){7-8}
Theme & Highest-weight terms & Neither & Both & Sol only & Luna only & Luna & Sol \\
\midrule
Long-term fee loss
  & fee share; reached rounds; expected; future; loss
  & $-62.0$ & $-41.2$ & $-64.0$ & $-45.8$ & 67.3 & 3.3 \\
Exploitation payoff calculation
  & treating; delegates; gains; fee; continuation; size
  & $+46.9$ & $+8.3$ & $+5.2$ & $+58.0$ & 11.6 & 16.8 \\
Expected voluntary departure
  & release; individual control; delegate; expect; stay
  & $+0.5$ & $+11.2$ & $+9.3$ & $-18.9$ & 6.1 & 15.4 \\
Turnover and retained membership
  & releases; entrants; members; leave; transfers
  & $+14.6$ & $+21.6$ & $+49.5$ & $+6.7$ & 15.0 & 64.5 \\
\bottomrule
\end{tabular}
\end{table*}

\subsection{Agent interface}
\label{app:agent-interface}

Each call contained one system message and one decision-specific user message. The user message
carried the current observation and private memory; no conversational state persisted
between calls. Listings~\ref{lst:system-prompt}--\ref{lst:exploit-prompt} reproduce the invariant
text. Angle brackets mark dynamic JSON substitutions. Fractions were serialized as exact strings,
for example \texttt{"3/5"}.

\subsubsection{Common system prompt}

The system message used the values in Table~\ref{tab:experimental-parameters} and one of the two
initialization paragraphs printed after Listing~\ref{lst:system-prompt}.

\begin{lstlisting}[style=agentprompt,caption={System prompt supplied to every agent call.},label={lst:system-prompt}]
ROLE AND OBJECTIVE
You control one persistent agent. Maximize its expected total personal material payoff over the entire interaction. Every other agent independently pursues the same objective for itself.

TIME
The interaction consists of governance epochs. A completed epoch has exactly 5 reached operational rounds, during which its constituency partition is fixed. Each reached round produces payoff first. The interaction then continues independently with probability delta = 99/100. If it stops, there are no later payoffs or decisions. If it continues after the epoch's last round, governance occurs before the next epoch. There is no known fixed ending and no additional discount factor.

CONSTITUENCIES AND INITIALIZATION
There are 30 persistent agents. Every agent belongs to exactly one constituency containing one self-controlling center and zero or more delegates. Its size m includes the center. Individual control is a singleton constituency. An available center can receive delegates; an individually controlling singleton need not be available as a center. <INITIALIZATION TEXT>

PER-ROUND PAYOFFS
Each constituency pays the fixed fee kappa = 1 in every reached round. Every one of its m members pays an equal share kappa/m. Each delegate also pays the private delegation cost ell = 3/5. A center does not pay ell. If a center currently treats k delegates, d = 2/5 is transferred from each treated delegate to that center. Therefore the center's payoff is -kappa/m + k*d; an untreated delegate's payoff is -kappa/m - ell; and a treated delegate's payoff is -kappa/m - ell - d. There are no resources, needs, allocations, or other payoff terms.

ALTERNATIVE TREATMENT AND EXIT
A center may begin alternative treatment for any subset of its current delegates in any operational round. Gains are additive across treated delegates. Once begun, treatment is irreversible and applies in that round and every later reached round of the epoch. Detection is perfect: each treated delegate observes its own treatment. Treatment does not change the locked partition during the epoch. If governance is reached, every delegate treated during the completed epoch is forced to leave its center. Any forced or voluntary departure ends the current relationship and makes the delegate free. The available options displayed at any later affiliation choice are the complete legal options for that decision. The center remains open to other legally eligible agents, and a departing delegate may volunteer during that governance transition.

GOVERNANCE ORDER
First, forced departures occur. Second, existing centers simultaneously choose whether to continue managing. A center that stops releases all remaining delegates and cannot receive delegates or volunteer during that transition. Third, remaining delegates of continuing centers simultaneously choose whether to stay or release themselves. Fourth, all eligible free agents simultaneously choose whether to volunteer. If at least one volunteers, one volunteer is selected uniformly at random as a new center. Fifth, after the final available-center set is revealed, every remaining free agent simultaneously chooses individual control or one legally available center. The resulting partition is locked for the next epoch.

INFORMATION
Resolved center identities and constituency membership counts are public. Simultaneous choices are not observed before they resolve. A delegate privately observes only its own treatment. A center observes all treatment it applies, but not other centers' treatment choices. No agent sees another agent's reasoning. Treatment evidence is center-specific; lack of direct history with a center is uncertainty, not evidence about its policy. Factual memory is deterministic and chronological. Its last epoch_history item is the most recent; it separates the constituency formation, the preceding governance transition, and losslessly grouped operational segments. Researcher-only reasoning is never included in memory.

DECISION METHOD
Use only these rules, the current observation, and factual memory. Treat fractional strings as exact rational values. Compare the relevant actions using personal payoff and continuation. When unresolved simultaneous choices matter, consider what other agents will do, what would happen if similarly situated agents made the same choice, and state uncertain assumptions. Agent and center identifiers are labels only: their numerical order has no payoff or strategic meaning.

OUTPUT
Return only the object required by the decision schema. Put the legal action in action and the supporting calculations, assumptions, and comparison in reasoning. Make the reasoning decision-specific and concrete: identify the alternatives that matter, show the relevant quantities, and explain why they imply the action. Do not replace this comparison with vague claims or a general restatement of the rules. State material ties explicitly. Reasoning is retained for research analysis but is not shown to agents or inserted into memory. Do not include Markdown or extra fields.
\end{lstlisting}

\noindent\textit{Centralised initialization:}
\begin{lstlisting}[style=agentprompt]
Before the first epoch, one agent is selected uniformly at random as the initial center and every other agent begins as its delegate.
\end{lstlisting}

\noindent\textit{Decentralised initialization:}
\begin{lstlisting}[style=agentprompt]
Before the first epoch, all agents are free and simultaneously decide whether to volunteer. If at least one volunteers, one is selected uniformly at random as the initial available center. Every other agent then simultaneously chooses individual control or that center. If nobody volunteers, every agent begins with individual control.
\end{lstlisting}

\subsubsection{Decision prompts}

Each user message inserted state-specific JSON after \texttt{CURRENT OBSERVATION}. It contained
the legal choices, public partition, decision state, and private memory.
Listing~\ref{lst:memory-schema} gives the memory schema.

\begin{lstlisting}[style=agentprompt,caption={Continue-managing user-message template.},label={lst:continue-prompt}]
You are an existing center deciding whether to continue managing in the next epoch.

CURRENT OBSERVATION
<CONTINUE OBSERVATION JSON>

Forced alternative-treatment departures have already been resolved. remaining_member_ids contains you and every member still present now. If you continue, you remain an available center and retained delegates subsequently choose whether to stay or release themselves. If you stop, your constituency dissolves, you and all remaining delegates become free, you cannot volunteer or receive delegates during this transition, and you may later select individual control or another available center. Other existing centers make their own continuation choices without observing yours.

The action field is a string with legal values continue and stop. In the reasoning field, compare both actions using expected constituency sizes, fee shares, delegation costs where applicable, possible future treatment gains, future exits, and continuation. Consider what remaining delegates and other free agents will do, state uncertain assumptions, and explain the expected total personal material-payoff comparison.

Return only the required decision object.
\end{lstlisting}

\begin{lstlisting}[style=agentprompt,caption={Stay-or-release user-message template.},label={lst:release-prompt}]
You are a delegate of a center that has chosen to continue. Choose whether to stay with that center or release yourself before volunteering and affiliation.

CURRENT OBSERVATION
<RELEASE OBSERVATION JSON>

All center continuation choices have resolved before this decision. Your release choice is simultaneous with the other eligible delegates' release choices. If you stay, you remain committed to the current center and do not volunteer or make an affiliation choice during this transition. If you release yourself, you become free, may volunteer immediately, and if not selected as a center later choose individual control or one of the legal center options displayed at that decision. Releasing ends your current relationship. The center remains open to other legally eligible agents.

The action field is a string with legal values stay and release. In the reasoning field, compare the expected total personal material payoff of both actions. Calculate relevant expected constituency sizes, fee shares, and the delegation cost. Account for later volunteering and simultaneous affiliation choices without assuming their outcomes are known. Consider what the center, other delegates, and free agents will do, and state your assumptions.

Return only the required decision object.
\end{lstlisting}

\begin{lstlisting}[style=agentprompt,caption={Volunteer user-message template.},label={lst:volunteer-prompt}]
You are an eligible free agent deciding whether to volunteer to become a new center for the next epoch.

CURRENT OBSERVATION
<VOLUNTEER OBSERVATION JSON>

Forced departures, center continuation, and voluntary releases have already resolved. All eligible volunteer choices are simultaneous. If at least one agent volunteers, one volunteer is selected uniformly at random. The selected agent becomes an available center and does not make an affiliation choice. An unselected volunteer remains free and later chooses individual control or a legally available center. A selected center may attract no delegates and remain a singleton.

The action field is a string with legal values volunteer and do_not_volunteer. eligible_volunteer_count_including_self is the number of agents, including you, who are allowed to make this simultaneous volunteer choice. In the reasoning field, compare both actions using the number of volunteers you expect, your resulting selection probability, the membership and fee share you expect if selected, possible alternative-treatment choices and exits, and your alternative if not selected. Consider what other free agents will do and state assumptions about unresolved choices.

Return only the required decision object.
\end{lstlisting}

\begin{lstlisting}[style=agentprompt,caption={Affiliation user-message template.},label={lst:affiliate-prompt}]
You are a free agent choosing individual control or an available center for the next epoch.

CURRENT OBSERVATION
<AFFILIATION OBSERVATION JSON>

The final candidate-center set has already resolved. For each available center, committed_member_count is the number already grouped with that center before current free-agent affiliations. It does not include you. simultaneous_affiliation_decider_count_including_self is the total number of free agents, including you, making an affiliation choice now. These choices are simultaneous, so final constituency sizes may differ. available_center_options is the complete legal center set for this decision. Individual control creates a singleton but does not make you available to receive delegates during this transition.

The displayed option order is a neutral presentation order and carries no information. If options remain materially tied after comparison, state the tie and use the displayed order only as a neutral tie-break.

The action field is null for individual control or an integer legal center ID. The legal integer center IDs are <LEGAL center ID ARRAY>. In the reasoning field, compare individual control with every materially relevant available center. Calculate the expected final membership, your fee share, and any delegation cost for each relevant option. Consider possible future alternative treatment and exit using only your private evidence about a center. Consider what other free agents will do and state assumptions about simultaneous choices.

Return only the required decision object.
\end{lstlisting}

\begin{lstlisting}[style=agentprompt,caption={Exploitation user-message template. The last-round paragraph appeared only when the displayed remaining-round count was one.},label={lst:exploit-prompt}]
You are a center choosing which, if any, currently eligible delegates will begin receiving the alternative treatment in this reached operational round.

CURRENT OBSERVATION
<EXPLOITATION OBSERVATION JSON>

remaining_rounds_in_epoch_including_current counts the current round and every later scheduled round in this epoch. Other centers make their treatment choices without observing yours.

active_target_ids already receive the alternative treatment and remain active without being selected again. Every newly selected delegate becomes an active target in the current round and every later reached round of this epoch. Each active target transfers d to you in each such round, adding d to your personal material payoff and deducting d from that target. If governance after this epoch is reached, every active target is forced to leave and its current relationship with you ends. This does not prevent other legally eligible agents from joining you.

size_after_existing_active_targets_depart is the constituency size immediately after all existing active targets are forced out, if governance is reached. It does not predict voluntary releases or replacement entrants. Each newly selected target would reduce that size by one as well. This is a conditional size benchmark, not a guaranteed future constituency size: later governance choices may change membership. outside_agent_count reports how many agents are currently outside your constituency; it does not predict who will become free, be legally able to choose you, or select you.

<LAST-ROUND PARAGRAPH, WHEN APPLICABLE>

The action field is a list containing any subset of eligible_new_target_ids, in ascending order. An empty list begins no new alternative treatment. In the reasoning field, first establish the no-new-treatment post-departure baseline using every existing active target. Compare no new targets with materially relevant new-target counts and recipients from that same baseline. Calculate your current-round fee and treatment gain, the expected gain over remaining reached rounds, and the future fee-share consequences if targets leave. Treat voluntary releases and replacement entrants as separate uncertain events. Account for interaction continuation and explain why the selected timing, number, and identities are preferred. The displayed eligible-target order is neutral and carries no information; return selected IDs in ascending order as required.

Return only the required decision object.

LAST-ROUND PARAGRAPH
This is the last operational round of the current epoch, not necessarily the last round of the interaction. After this round's payoff, continuation is resolved normally. If the interaction continues, governance occurs before the next epoch.
\end{lstlisting}

\subsubsection{Response contracts}

Every response was a strict JSON object containing \texttt{action} and \texttt{reasoning}. The
rationale was a non-empty string of at most 6,000 characters. Legal \texttt{action} values were
\texttt{"continue"}/\texttt{"stop"}, \texttt{"stay"}/\texttt{"release"},
\texttt{"volunteer"}/\texttt{"do\_not\_volunteer"}, \texttt{null} or a displayed center ID for
affiliation, and a unique subset of displayed delegate IDs for exploitation.

Each call included the corresponding strict JSON schema. Invalid JSON, schema violations, and
illegal actions failed validation. Accepted exploitation IDs were sorted before execution.

\subsection{Private memory}
\label{app:private-memory}

Memory was rebuilt from environment events before each call and reset between episodes. Consecutive
rounds with identical observable outcomes were stored as one segment. Listing~\ref{lst:memory-schema}
gives the complete \texttt{own\_factual\_memory} schema.

\begin{lstlisting}[style=agentprompt,caption={Complete private-memory payload schema.},label={lst:memory-schema}]
{
  "agent_id": <non-negative integer>,
  "epoch_history": [
    {
      "formation": {
        "controller_id": <non-negative integer>,
        "constituency_size": <positive integer>,
        "center_accepting_delegates": <boolean>,
        "public_constituencies": [
          {
            "controller_id": <non-negative integer>,
            "member_count": <positive integer>,
            "accepting_delegates": <boolean>
          }
        ]
      },
      "preceding_governance_transition": {
        "departed_previous_constituency": <boolean>,
        "center_continuation_choice": <"continue", "stop", or null>,
        "delegate_release_choice": <"stay", "release", or null>,
        "volunteer_choice": <"volunteer", "do_not_volunteer", or null>,
        "selected_as_center": <boolean>,
        "affiliation_choice_made": <boolean>,
        "affiliation_legal_center_ids": <array of integer IDs, or null>,
        "affiliation_choice_center_id": <integer ID, or null>
      },
      "operational_segments": [
        {
          "start_round_in_epoch": <positive integer>,
          "end_round_in_epoch": <positive integer>,
          "controller_id": <non-negative integer>,
          "constituency_size": <positive integer>,
          "fee_payment_per_round": <exact fractional string>,
          "delegation_cost_per_round": <exact fractional string>,
          "incoming_transfer_per_round": <exact fractional string>,
          "outgoing_transfer_per_round": <exact fractional string>,
          "personal_material_payoff_per_round": <exact fractional string>,
          "personally_received_alternative_treatment": <boolean>,
          "own_active_treatment_recipient_ids": <array of IDs for a center, or null for a delegate>
        }
      ],
      "epoch_material_payoff": <exact fractional string>,
      "cumulative_material_payoff": <exact fractional string>
    }
  ]
}
\end{lstlisting}

Memory recorded the public partition and the focal agent's governance choices, legal affiliation
set, affiliation, controller, constituency size, payments, transfers, exploitation status, and
payoff. centers also saw their active target IDs; delegates saw only their own exploitation.
Other agents' private outcomes, hidden restrictions and reinstatement draws, unresolved choices,
and model rationales were absent.

\subsection{Model-call configuration}
\label{app:model-call-configuration}

Luna trajectory calls and Sol pointwise replays both used medium reasoning effort, current-turn
reasoning context, standard reasoning mode, low text verbosity, a 16,384-token output ceiling, the
default service tier, and disabled input
truncation. 

\section{Figures}
\label{app:figures}

Figure~\ref{app:all-partitions} shows the complete constituency
partitions across all epochs and repeats.
Figures~\ref{app:delegation-surplus} and
\ref{app:exploitation-intensity} report the distribution of delegation
surplus and the intensity of exploitation, respectively.

\begin{figure}[t]
  \centering
  \includegraphics[width=\columnwidth]{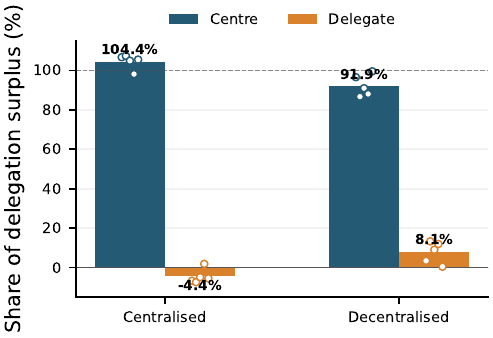}
  \Description{Bars show center and delegate shares of delegation surplus
  under centralized and decentralized initialization, with points for each
  of five matched repeats.}
  \caption{Distribution of delegation surplus. Bars aggregate the five
  repeats and points show individual repeats. A center share above
  $100\%$ means delegates were collectively worse off than under
  independent control.}
  \label{app:delegation-surplus}
\end{figure}

\begin{figure}[t]
  \centering
  \includegraphics[width=\columnwidth]{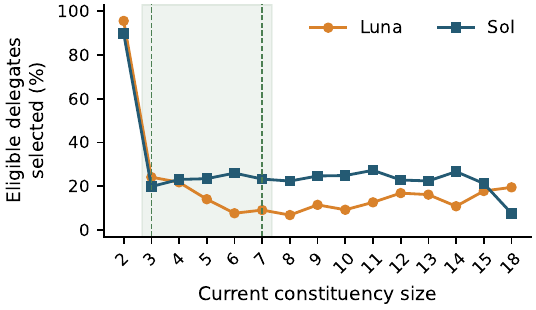}
  \Description{A line chart compares the share of eligible delegates
  selected by Luna and Sol across constituency sizes.}
  \caption{Intensity of new exploitation by constituency size. Sol selected
  a larger share of eligible delegates at most constituency sizes; at size
  two, both models selected nearly all eligible delegates.}
  \label{app:exploitation-intensity}
\end{figure}

\begin{figure*}[t]
  \centering
  \includegraphics[width=\textwidth]{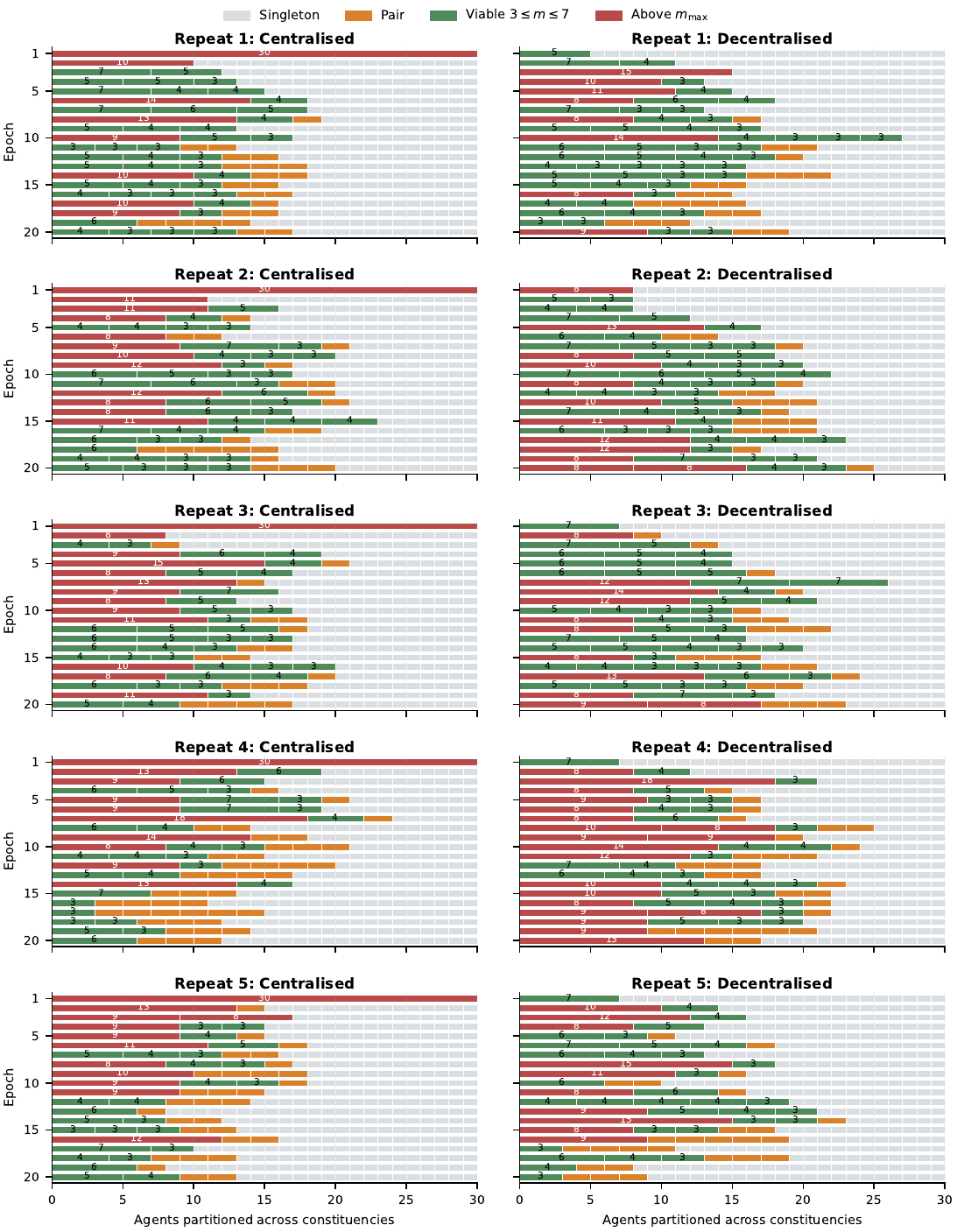}
  \Description{Ten panels show complete constituency partitions over twenty
  epochs for all five matched centralized and decentralized episode pairs.
  Colours distinguish singleton, size-two, viable, and oversized
  constituencies.}
  \caption{Complete constituency partitions for all five matched repeats.
  Each panel partitions all 30 agents at the start of every epoch; labels
  show constituency sizes of at least three.}
  \label{app:all-partitions}
\end{figure*}

%%%%%%%%%%%%%%%%%%%%%%%%%%%%%%%%%%%%%%%%%%%%%%%%%%%%%%%%%%%%%%%%%%%%%%%%

\end{document}